\documentclass[twocolumn]{autart}    % Enable this line and disable the 
\usepackage{graphicx}          % Include this line if your 
\usepackage{amsmath} % assumes amsmath package installed
\usepackage{amssymb}  % assumes amsmath package installed
\usepackage{natbib}
\setcitestyle{authoryear,open={(},close={)}}
\usepackage{color}
\usepackage{mathtools}
\usepackage{url}
\usepackage{booktabs}

\newcommand{\longthmtitle}[1]{\mbox{}\textup{\textbf{(#1):}}}

\newcommand{\R}{\mathbb{R}}

\newcommand{\Rpluseq}{\mathbb{R}_{\geq 0}}

\newcommand{\norm}[1]{\left\Vert #1 \right\Vert}
\newcommand{\x}{\mathbf{x}}

\newcommand{\W}{\mathbf{W}}
\renewcommand{\c}{\mathbf{c}}

\renewcommand{\u}{\mathbf{u}}

\newcommand{\y}{\mathbf{y}}
\newcommand{\diag}{\mathrm{diag}}

\renewcommand{\v}{\mathbf{v}}

\newcommand{\J}{\mathbf{J}}
\newcommand{\argmin}{\operatorname{argmin}}

\newcommand{\afree}{a_{\mathrm{free}}}

\newcommand{\map}[3]{#1:#2 \rightarrow #3}

\graphicspath{{figures/}}

\newtheorem{theorem}{Theorem} % [2026-07-20 edit: sequential numbering per Automatica style]
\newtheorem{definition}[theorem]{Definition}
\newtheorem{assumption}[theorem]{Assumption}
\newtheorem{lemma}[theorem]{Lemma}

\newtheorem{proposition}[theorem]{Proposition}

\begin{document}

\begin{frontmatter}

\title{Dynamics-matched Physical Reservoir Computing for Undersensed Traffic Prediction\thanksref{footnoteinfo}}

\runtitle{Reservoir Computing for Undersensed Traffic Prediction}

\thanks[footnoteinfo]{This paper was not presented at any IFAC
meeting. This work was supported by Toyota InfoTech Labs and by the
Natural Sciences and Engineering Research Council of Canada (NSERC)
through the Canada Research Chairs Program. Corresponding author
S.~L.~Smith.}

\author[UW]{Michael McCreesh}\ead{michael.mccreesh@uwaterloo.ca},
\author[Toyota]{Rohit Gupta}\ead{rohit.gupta@toyota.com},
\author[UW]{Stephen L. Smith}\ead{stephen.smith@uwaterloo.ca}

\address[UW]{Department of Electrical and Computer Engineering,
University of Waterloo, Waterloo, ON, Canada}
\address[Toyota]{Toyota InfoTech Labs, Mountain View, CA, USA}
          
                             % keyword list or with the 
                                          % help of the Automatica 
                                          % keyword wizard

\begin{abstract}                          % Abstract of not more than 200 words.
Machine learning methods are increasingly used for traffic prediction in applications such as autonomous driving. Such predictions must be both highly accurate and immediately available, making methods with low computational costs and fast training times of interest. One such method is reservoir computing, in which the rich dynamics of a nonlinear system serves as a computational substrate and only a linear readout vector is trained. In this work we use a traffic network as the reservoir for predicting the behavior of an undersensed traffic network. This matching of the highly nonlinear dynamics allows for similar encoding between the behaviors of the reservoir and target network, enabling a more direct prediction. We show that a reservoir governed by the Improved Intelligent Driver Model (IIDM) satisfies the echo state property for a class of slowly-varying inputs. Through simulations we show that the echo state property likely holds for a larger class of inputs, and that the IIDM reservoir computer (IIDM-RC) accurately predicts an undersensed vehicle network governed by varying car-following models. We also compare with echo state networks (ESNs) and Long Short-Term Memory (LSTM) networks, finding improvements using IIDM-RC in both prediction accuracy and training time.
\end{abstract}

\end{frontmatter}

\section{Introduction}
Predicting traffic flow is important for a range of applications, including autonomous driving~\citep{YY-PZ-CD-HW-HW-DO-YHC:26}, advanced driver assistance systems~\citep{SA-MRCQ-WA-AS-SN-HA:24}, and congestion mitigation in transportation networks~\citep{JWL-ETAL:25,JW-YZ-KL-QX:23}. In these settings, accurate short-term prediction of vehicle behavior can improve safety, efficiency, and coordination. Since traffic is a nonlinear dynamical system with strongly coupled vehicle interactions, it remains challenging to develop prediction methods that are both accurate and computationally practical. In addition to the complexity of the traffic dynamics, the full state of a set of vehicles is not typically available. Instead, one is usually able to observe a subset of information, such as vehicle velocities or spacing, based on a mix of connected vehicles and limited sensors. As such, the problem of predicting vehicle platoon behavior from noisy~\citep{DW-KM-KHJ-MB:25,YM-ZF-IN-EF-TMLK-AB:25} or incomplete~\citep{MB-MA-JL-KHJ:21,JBC-MC:19} data is important, and is closely related to traffic state estimation~\citep{TS-AMB-TK-YA:17}.

Many widely used machine learning approaches for traffic prediction, including deep recurrent neural networks, Long Short-Term Memory networks, and transformers, can provide strong predictive performance but often require substantial computational resources and long training times. Although inference with a fixed trained model may be feasible in real time, traffic conditions and the available sensed information continually change, and maintaining accurate predictions requires retraining as conditions evolve. This is a significant limitation in applications such as autonomous driving and advanced driver assistance systems, where prediction must be carried out in real time on-board the vehicle or near the vehicle at the network edge, where compute resources are limited, and where learning must operate alongside real-time control~\citep{AAM:23}. Reservoir computing~\citep{HJ:10}, in which a fixed nonlinear dynamical system (the reservoir) provides the computational substrate and only an output layer is trained, is attractive since it reduces learning to a linear regression problem. As a result, it offers lower computational cost and faster retraining than methods such as LSTMs while still capturing the nonlinear dynamical structure of traffic flow~\citep{GT-TY-JBH-RN-NK-ST-HN-DN-AH:19}.

In physical reservoir computing, the reservoir is not implemented as a simulated recurrent network, but is instead provided by the dynamics of an existing physical system~\citep{KN:20}. Vehicular traffic is a natural candidate for use as a reservoir for the problem of traffic prediction: its dynamics are high-dimensional and nonlinear, and relevant quantities such as vehicle speed, position, and spacing are already measured by existing sensing and communication systems~\citep{KFC-FY-AA-FI:25,HA-HC:19,RF-TN-HA:25}. In addition, this setup matches the dynamics between the reservoir and the system to be predicted. A benefit of this approach is that the set of possible behaviors of the reservoir and that of the target system will be similar, making it easier to relate the two with the single readout vector. 

In this work we study the use of a vehicular platoon as a reservoir, particularly for application to the prediction of an undersensed vehicle network. The first goal of the paper is to establish the mathematical validity of using a vehicle platoon as a reservoir. To do so we simulate vehicle behavior using a variation of the Intelligent Driver Model~\citep{MT-AH-DH:00} introduced in Section~\ref{sec:IIDM}, and illustrate that the echo state property holds for a class of slowly-varying signals. Second, we provide a series of simulations illustrating the use of a vehicular reservoir can provide a fast and accurate prediction of the undersensed vehicle platoon's state.

%%-----------------------------------------------------------------------
\section{Car-Following Models and the Improved Intelligent Driver Model}\label{sec:IIDM}
%%-----------------------------------------------------------------------
Car-following models are representations of microscopic traffic flow, describing the behavior of each of a group of $N$ vehicles based on the environment and the surrounding vehicles. In this work we focus on longitudinal models, which describe motion along the lane, rather than lateral models, which describe lane changes. Longitudinal car-following models date back to the $1950$s, including Pipes' model~\citep{LAP:53}, Newell's model~\citep{GFN:61}, the Optimal Velocity Model (OVM)~\citep{MB-KH-AN-AS-YS:95}, and Gipps' model~\citep{PGG:81}. In this paper we consider a variation of the Intelligent Driver Model (IDM)~\citep{MT-AH-DH:00}.

The IDM is a continuous-time deterministic car-following model that defines the acceleration of the vehicle in terms of four main quantities: i) its current speed, $v$; ii) the distance to the vehicle in front, $s$; iii) the difference in speed between the vehicle and the leader, $\Delta v$; and iv) the desired speed, $v_0$. For the original IDM formulation the acceleration function is given by

\begin{align}\label{eq:IDM}
  \dot{v} = a\left(1 - \left(\frac{v}{v_0}\right)^\delta - \left(\frac{s^*(v,\Delta v)}{s}\right)^2   \right),
\end{align}
where $s^*(\cdot,\cdot)$ is the desired gap given by
\begin{align}\label{eq:desired_gap}
  s^*(v,\Delta v) = s_0 + \max\left(0,vT + \frac{v\Delta v}{2\sqrt{ab}}\right),
\end{align}
and $\Delta v = v - v_\ell$, where $v_{\ell}$ is the velocity of the leader vehicle. The behavior of the model depends only on physically interpretable parameters: the desired speed $v_0$ ($\mathrm{m/s}$), the minimum distance gap $s_0$ ($\mathrm{m}$), the desired time gap $T$ ($\mathrm{s}$), the maximum acceleration $a$ ($\mathrm{m/s^2}$), the maximum comfortable deceleration $b$ ($\mathrm{m/s^2}$), and the acceleration exponent $\delta$, which dictates how close vehicles get before decelerating.\footnote{In a multi-vehicle network we will denote the parameters of the $i$'th vehicle by $p^{(i)}$ for $p \in \{v_0,s_0,T,a,b,\delta\}$. Typical values for highway driving given in~\citep{MT-AK:13} are $v_0 = 35$, $s_0 = 2$, $T = 1$, $\delta = 4$, $a = 1$, and $b = 1.5$.}

The most important part of the IDM in comparison to other models is the `intelligent braking strategy', which is given by the third term in the desired spacing function~\eqref{eq:desired_gap}~\citep{MT-AK:13}. This term is derived from the concept of kinematic braking, and ensures that the vehicle maintains a sufficient safety margin when approaching a slower vehicle~\citep{SZ-SZ-JT-RJ-HMZ:25}. Due to the inclusion of the relative speed between vehicles the model has smooth and anticipative braking not exceeding the comfortable braking value, $b$, except in dangerous situations (such as a car cutting in front with too small a gap). Due to this safety property the IDM has been used extensively in Advanced Driver Assistance Systems, such as adaptive cruise control, where safety is critical~\citep{MT-AK:13,YL-HW-WW-LX-SL-XW:17}, as well as in connected cruise control~\citep{GO:16}.

While the IDM typically provides realistic simulations of traffic behavior, it can exhibit excessive braking when the velocity exceeds the desired velocity or the gap is far below desired, steady-state gaps larger than observed near the desired speed (with not all vehicles reaching it)~\citep{MT-AK:13}, and negative velocities for certain initial conditions and parameters~\citep{SA-AB-MTC-XG-AH-NK-AK-STM-BP-YY:22}. Multiple extensions address these issues, see~\cite[Section 3]{SZ-SZ-JT-RJ-HMZ:25} for a review. Going forward we will use an updated acceleration function, combining approaches taken in~\citep{MT-AK:13} and~\citep{SA-AB-MTC-XG-AH-NK-AK-STM-BP-YY:22}, called the Improved Intelligent Driver Model (IIDM).

The IIDM is a piecewise formulation of the IDM in which the state-space is broken into four regimes, based on the ratios $\frac{v}{v_0}$ and $\frac{s^*(v,\Delta v)}{s}$. We first define a free acceleration term, dependent on the difference between the velocity and desired velocity, as follows:
\begin{align}\label{eq:free_accel}
    \afree(v) = \begin{cases}
    a\!\left[1 - \left(\dfrac{v}{v_0}\right)^{\!\delta}\right] & v \leq v_0,\\[8pt]
    -b\!\left[1 - \left(\dfrac{v_0}{v}\right)^{\!a\delta/b}\right] & v > v_0.
  \end{cases}
\end{align}
Using the desired gap function~\eqref{eq:desired_gap}, and defining the variable $z = \frac{s^*(v,\Delta v)}{s}$, we define the full acceleration function as
\begin{align}\label{eq:IIDM}
  \dot{v} = \begin{cases}
    a(1-z^2) & v \leq v_0, z \geq 1 \\
    \afree(v)(1- z^{2a/\afree(v)}) & v \leq v_0, z < 1 \\
    \afree(v) + a(1-z^2) & v > v_0, z\geq 1 \\
    \afree(v)  &  v > v_0, z < 1 \\
    0 & v = 0, s < s_0
  \end{cases}
\end{align}
The final case takes precedence whenever it applies. The first four cases are defined as in~\citep{MT-AK:13}, and result in the vehicles reaching their desired speed and more realistic deceleration profiles. Meanwhile, setting the acceleration to $0$ when the velocity is zero and the gap is small is based on~\citep{SA-AB-MTC-XG-AH-NK-AK-STM-BP-YY:22} and prevents the case of negative velocities. While this introduces a discontinuity into the system, an argument following the proof of Theorem 5.15 of~\citep{SA-AB-MTC-XG-AH-NK-AK-STM-BP-YY:22} shows that this does not impact existence or uniqueness of solutions.

The model as written above is for a single vehicle following a leader whose velocity appears only in the calculation of $\Delta v$. However, it is important to note that this can be immediately extended to a set of $N$ vehicles, where the preceding vehicle is considered as the leader. The overall leader can then be considered as having a velocity given by a provided function, or by treating it as a vehicle with an arbitrarily large gap in front of it, where the system input is given as the desired velocity.

% ----------------------------------------------------------------------
\section{Reservoir Computing}\label{sec:reservoir_computing}
% ---------------------------------------------------------------------- 

Here we provide an explanation of the mathematical basis of reservoir computing. Reservoir computing is a machine learning framework in which a fixed dynamical system (the reservoir), exhibiting significantly richer behavior than the target system, is used to predict the target outputs; this richness allows the target trajectories to be written as a trained linear function of the reservoir states. We define the basic framework as follows. Assume we have a dynamical system defined by
\begin{align}\label{eq:RC_prediction_system}
  \dot{\v}(t) = f(\v(t),\u(t)) \qquad \y(t) = g(\v(t))
\end{align}
where $\v \in \R^n$ is the system state, $\y \in \R^\ell$ is the system output, $\u(t) \in \R^m$ is
the input and both the driving function $\map{f}{\R^n \times \R^m}{\R^n}$ and output function $\map{g}{\R^n}{\R^{\ell}}$ are unknown. We define a reservoir as a dynamical system 
\begin{align}\label{eq:RC_reservoir}
  \dot{\x}(t) = F(\x(t),\c(t)), 
\end{align}
where $\x \in X \subseteq \R^N$ is the internal state of the
reservoir and $\c(t) \in C \subseteq \R^m$ is a function of $\y(t)$ and $\u(t)$. We assume that $N \gg n$ and that $\map{F}{X \times C}{\R^N}$ is nonlinear. The goal of the reservoir computing
framework is to use the activity of the high-dimensional reservoir
dynamical system~\eqref{eq:RC_reservoir} to estimate the outputs of
the unknown system~\eqref{eq:RC_prediction_system}. The output estimate is defined as
\begin{align*}
  \tilde{\y}(t) = \J_{\text{out}}\x(t),
\end{align*}
where $\J_{\text{out}} \in \R^{\ell \times N}$ is the output matrix
trained to achieve an accurate estimate of the unknown system.

To define the echo state property, we introduce the following
notations. Let $X^{+\infty} \coloneqq \{\x^{+\infty} = \{\x(t)\}_{t = 0}^{\infty}
  ~|~ \x(t) \in X, ~ \forall t\geq 0 \}$ and
  $C^{+\infty} \coloneqq \{\c^{+\infty} = \{\c(t)\}_{t = 0}^{\infty}
  ~|~ \c(t) \in C, ~ \forall t\geq 0 \}$ denote sets of right infinite
  state and input sequences. A right infinite state sequence
  $\x^{+\infty}$ is \emph{compatible} with input state sequence
  $\c^{+\infty}$ when $\dot{\x} = F(\x(t),\c(t))$ for all
  $t \geq 0$.

\begin{definition}\label{def:echo-state-property}
  A reservoir $\map{F}{X \times C}{X}$ defined on compact sets
    $X$ and $C$ satisfies the echo state property with respect to $C$
    if and only if there exists a function
    $\map{\delta}{\Rpluseq}{\Rpluseq}$ such that for any right infinite input sequence
    $\c^{+\infty} \in C^{+\infty}$ and any two right infinite state
    vector sequences $\x_1^{+\infty}, \x_2^{+\infty} \in X^{+\infty}$
    compatible with $\c^{+\infty}$, it holds that
    $\norm{\x_1(t)-\x_2(t)} \leq \delta(t)$ for all $t \geq 0$, where
    $\lim_{t \to \infty} \delta(t) = 0$~\citep{IBY-HJ-SJK:12}.
\end{definition}

This property can be thought of in terms of convergence of trajectories of the system. In particular, any two trajectories of the system under the same input must converge to each other, which relates to to contraction analysis~\citep{WL-JJES:98,DA:02,HT-SJC-JJES:21} and the (equivalent on a compact set) concepts of incremental stability and convergent systems applied to a controlled system~\citep{GG-AM-MAP-JJF-SA-MD:26}. For classical reservoir computing, in which the function $F$ denotes a recurrent neural network with a sigmoidal activation function~\citep{HJ:10}, conditions for the ESP depend on both the activation function and the weight matrix $\W$ that defines the structure of the RNN~\citep{IBY-HJ-SJK:12}. For a more general reservoir computing setup conditions must be checked on a case-by-case basis.

What remains is to train the reservoir computing framework to get an
accurate estimate $\tilde{\y}$ of the output $\y$. The defining
feature of reservoir computing is that we only train the output matrix
$\J_{\text{out}}$, rather than the internal parameters of the
reservoir. To train, a signal $\{\c_T(t)\}$ drives both the system~\eqref{eq:RC_prediction_system} and the reservoir~\eqref{eq:RC_reservoir}; the resulting output and reservoir-state timeseries are compiled into matrices $\mathbf{Y}_T$ and $\mathbf{X}_T$, and the output matrix $\J_{\text{out}}$ is computed using a Tikhonov regularization,
\begin{align}\label{eq:Tikhonov_regularization}
  \argmin_{\J_{\text{out}}} \norm{\mathbf{Y}_T - \J_{\text{out}}\mathbf{X}_T}^2 +
  \norm{\beta\J_{\text{out}}}^2, 
\end{align}
where $\beta > 0$ is a regularization parameter. As this single linear regression is used for training the RC, it is significantly faster than methods such as backpropagation.

\section{Mathematical Guarantees for Vehicular Reservoirs}

In this section we consider the problem of using a vehicular reservoir for traffic prediction. We will consider an $N$-vehicle platoon and define the velocity of the vehicles as the state of the reservoir. The input to the reservoir will then be either the desired or actual velocity of the lead vehicle. In Figure~\ref{fig:iidm_rc} we provide a schematic of this setup, particularly for the problem of predicting a distinct vehicle network. For the simulation of such networks we will use the IIDM dynamics, and refer to this setup as IIDM-RC. For this approach to be feasible we need to show that the dynamics satisfy the ESP, which we will show is true for a class of slowly-varying signals. While this class of signals only includes signals with lower accelerations than typically observed, the simulations in Section~\ref{sec:simulations} suggest that it is also valid for physically realistic signals.

\begin{figure}
\centering
  \includegraphics[width = 0.75\linewidth]{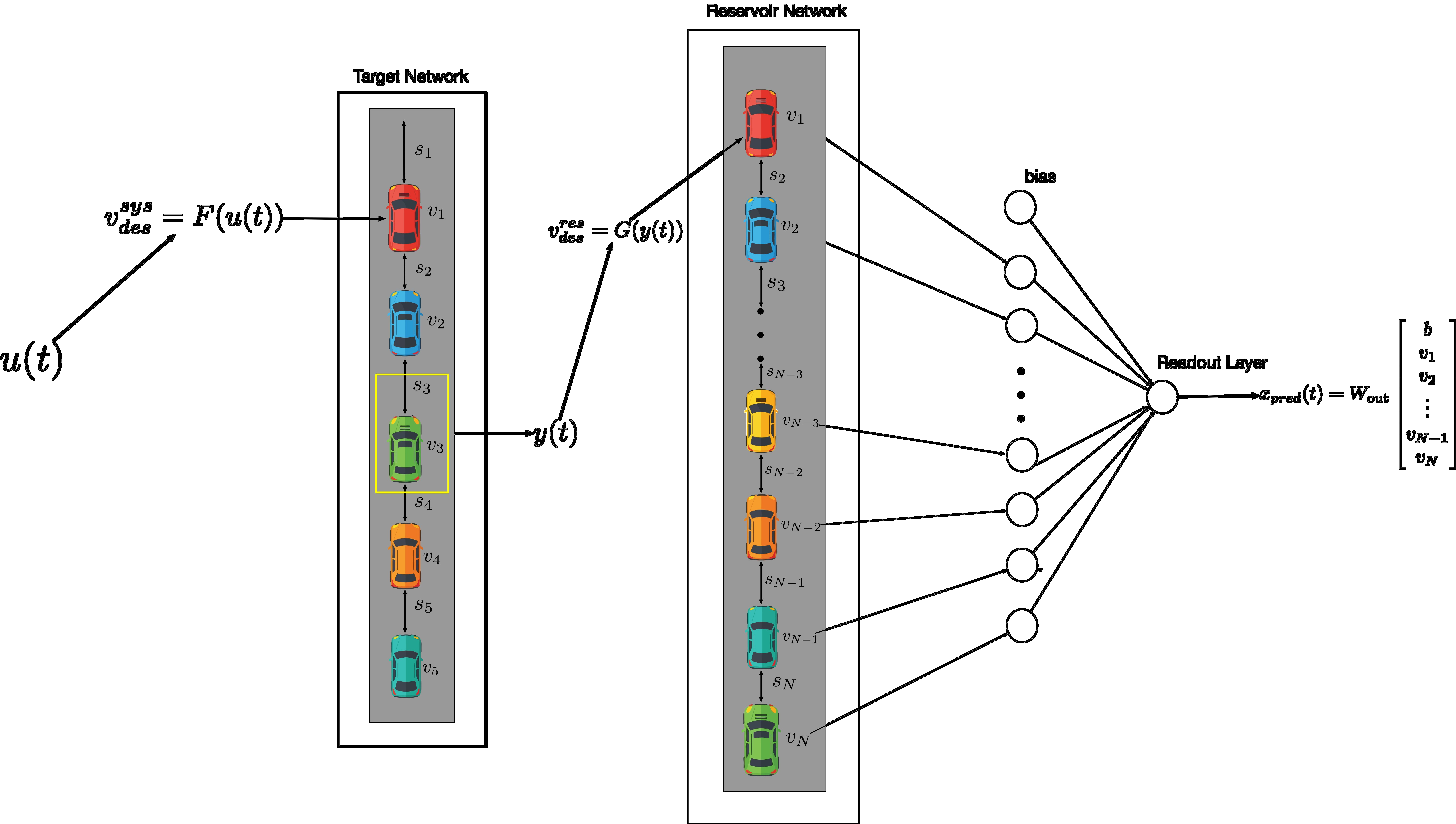} % [2026-07-21 shrink: was 0.9]
  \caption{The schematic for using a vehicular network as a reservoir for predicting an undersensed network through reservoir computing. The sensed values from the target network are mapped to be the desired velocity in the reservoir network and the velocities of the reservoir network are used as its state.}\label{fig:iidm_rc}
\end{figure}

To show that an IIDM network of $N$ vehicles satisfies Definition~\ref{def:echo-state-property} we will first show that it holds for a single following vehicle before an extension to an $N$-vehicle platoon. We will denote the dynamics for the single IIDM follower as follows. Let $x = (v,s)$ be the state of the system and $u(t)$ be the input to the system, which represents the velocity of the leading vehicle. The dynamics of the network are then given by
\begin{align}\label{eq:IIDM_follower}
  \dot{x} = F(x,u) = (A(v,s,u),u-v),
\end{align}
where $A(v,s,u)$ is the IIDM acceleration function~\eqref{eq:IIDM}. Throughout this section we will assume $v_0$ is constant and consider the class of input signals
\begin{align}\label{eq:slowly_varying_input_signal_class}
    \mathcal{U}_\rho = &\{\map{u}{\Rpluseq}{[u_{\min},u_{\max}]}~|~\\ \notag
    &\text{locally absolutely continuous and } |\dot{u}| \leq \rho \text{ a.e.}\},
\end{align}
where $u_{\max} < v_{0}$ and $u_{\min} > 2T\sqrt{ab}$. We first introduce the following definition.

\begin{definition}
  A trajectory $x(t)$ of the IIDM dynamics satisfies the \emph{gap-increasing condition} from time $T_0$ if whenever $t \geq T_0$ and $s(t) \leq s_0$, the velocity gap $w(t) = u(t) - v(t) > 0$.
\end{definition}
This condition is used for guaranteeing a lower bound on the trajectories of the IIDM. In the remainder of this section we take the gap-increasing condition as a standing assumption, together with $w(0) > 0$ and $s(0) > s_0$, so that $T_0 = 0$. The assumption is supported by the following parametric bound:
\begin{align}\label{eq:gap_increasing_regime_A}
  a\left(\frac{2u_{\min}T}{s_0} + \left(\frac{u_{\min}T}{s_0}\right)^2\right) > b,
\end{align}
under which any vehicle at a gap $s \leq s_0$ that is not already opening ($w \leq 0$) decelerates at a rate exceeding $b$, driving the velocity gap positive. This bound holds for the standard highway parameters of the IIDM given above, along with a region surrounding them due to continuity.

\begin{proposition}\label{prop:X_invariance}
  Consider the IIDM dynamics~\eqref{eq:IIDM_follower} and suppose that the gap-increasing condition is satisfied. Fix $\bar{S} > s_0$. If $s(0) \in (s_0,\bar{S}]$ and $v(0) \in [0,v_0]$, then all trajectories of~\eqref{eq:IIDM_follower} remain in the set $X \subset \R^2$ defined as
  \begin{align*}
    X = \{ (v,s) ~|~ 0 \leq v \leq v_0, s_0 \leq s \leq s_{\max} \},
  \end{align*}
  where $s_{\max} = \max(\bar{S},\frac{S_1}{z_0}) + \frac{u^2_{\max}}{2\alpha_0}$ for any fixed $z_0 \in (0,1)$, $S_1 = s_0 + u_{\max}T$, and $\alpha_0 = \afree(u_{\max})(1-z_0^{2a/\afree(u_{\max})})$.
\end{proposition}
\begin{pf}
  We illustrate that all trajectories remain in the set by considering behavior on each boundary face: (i) For the face $v = v_0$ we compute $\afree(v_0) = 0$. Then, examining the subcases of $A$, for $z \geq 1$ we get $A(v_0,s,u) = a(1-z^2) \leq 0$. When $z < 1$ we have $A(v_0,s,u) = \afree(v_0)(1-z^{2a/\afree(v_0)})$. Since $\afree(v_0) = 0$, this is not well-defined in the exponent. Considering the value as the limit $\lim_{v \to v_0} A(v,s,u)$, we get that it is equal to $0$. As such, for all cases when $v = v_0$, we have $\dot{v} \leq 0$ and the vector field points into $X$.
    (ii) For the face of $v = 0$ we consider when $s \geq s_0$ and $s < s_0$. When $s < s_0$ we have that $\dot{v} = 0$ directly from the model. When $s \geq s_0$ we compute $s^*(0,u) = s_0 + \max(0,0) = s_0$. As such $z \leq 1$. Since $v \leq v_0$, when $z = 1$ we get $A = \afree(0)(1-z^2) = 0$. When $z < 1$ we have $A = \afree(0)(1-z^{2a/\afree(0)})$ which is nonnegative since $\afree(0) = a$ and $z \leq 1$. Therefore, $\dot{v} \geq 0$ when $v = 0$, meaning the vector field points inwards.
    (iii) We illustrate that $s = s_0$ is an invariant positive lower bound for $s$ as follows. By the gap-increasing condition, $\dot{s}(t) > 0$ whenever $s(t) \leq s_0$. Since $s(t)$ is a continuous and differentiable function, applying Proposition~\ref{prop:lower_bound_invariance} guarantees $s(t) > s_0$ for all $t > 0$. Therefore this face of the set is invariant.
    (iv) For the face $s \leq s_{\max}$ we are interested in phases when $\dot{s} = u - v > 0$, that is the gap is increasing. During phases with $u > v$ we have that the kinematic term in $s^*$ is non-positive, meaning that $s^*(v,u) \leq s_0 + vT \leq s_0 + u_{\max}T = S_1$. Now, fix some $z_0 \in (0,1)$ and consider the region $s \geq S_1/z_0$. In this region we get
    \begin{align*}
        z = \frac{s^*}{s} \leq \frac{S_1}{S_1/z_0} = z_0 < 1.
    \end{align*}
    As such, the acceleration function is given by $A(v,s,u) = \afree(v)(1-z^{2a/\afree(v)}) \geq \afree(u_{\max})(1-z_0^{2a/\afree(u_{\max})}) = \alpha_0 > 0$, where the inequality uses $z \leq z_0 < 1$, that $\afree$ is decreasing with $v < u \leq u_{\max}$ during gap-increasing phases, and that $\alpha \mapsto \alpha(1-z_0^{2a/\alpha})$ is increasing in $\alpha$ (writing $z_0^{2a/\alpha} = e^{-\beta/\alpha}$ with $\beta = 2a\ln(1/z_0)$, its derivative is $1-(1+\beta/\alpha)e^{-\beta/\alpha} > 0$). Since $u(t) \leq u_{\max}$ and $\dot{v} \geq \alpha_0$ we get
    \begin{align*}
        u(t) - v(t) \leq u_{\max} - v(\tau) - \alpha_0(t-\tau),
    \end{align*}
    which reaches $0$ within time $\Delta t \leq u_{\max}/\alpha_0$. After this time we get $\dot{s} \leq 0$. The gap gain above $S_1/z_0$ is given by
    \begin{align*}
        s(t) - s(\tau) \leq \int_0^{u_{\max}/\alpha_0}(u_{\max}-\alpha_0t)dt = \frac{u^2_{\max}}{2\alpha_0}.
    \end{align*}
    As such, we get that $s(t) \leq \max(s(0),S_1/z_0) + \frac{u_{\max}^2}{2\alpha_0} \leq s_{\max}$.
  As we have shown the vector field is bounded in all directions all trajectories remain in the set $X$.
\end{pf}

To prove the ESP on the set $X \times \mathcal{U}_\rho$ we use a series of intermediate results. The first set of results illustrates that for any input $u(t) \in \mathcal{U}_\rho$ all trajectories of~\eqref{eq:IIDM_follower} will converge to an ellipsoid around the moving equilibrium point. The second set of results shows that any two trajectories inside this ellipsoid contract to each other, which proves the ESP for the single follower. 

In order to prove these results we provide the following definitions. For a frozen input $u \in [u_{\min},u_{\max}]$ the unique equilibrium point is $\bar{x}(u) = (u,\bar{s}(u))$, with $\bar{s}(u) = s_0 + uT$. The equilibrium manifold is the set $\mathcal{E} = \{\bar{x}(u)~|~u \in [u_{\min},u_{\max}]\}$. For some $ c > 0$ the moving ellipsoid and corresponding graph set are
\begin{align*}
    \Omega_{c}(u) &= \{x~|~ \norm{x - \bar{x}(u)}_P^2 \leq c \} \\
    \mathcal{W}_c &= \{(x,u) ~|~ u \in [u_{\min},u_{\max}],~x \in \Omega_c(u)\},
\end{align*}
where $P$ is a positive definite metric. We denote $K(x,u) = \partial A/\partial s$, $D(x,u) = -\partial A/\partial v$, and let $J(x,u)$ be the Jacobian of $F(x,u)$. We make the following assumption on the existence of a contracting metric $P$ on the ellipsoid.

\begin{assumption}\label{assumption:contraction_metric}
    Suppose that there exists $P = P^\top \succ 0$, $\bar{c} > 0$ and $\lambda > 0$ such that
    \begin{align}\label{eq:certificate}
        S(x,u) \coloneqq PJ(x,u) + J(x,u)^\top P \preceq -2\lambda P
    \end{align}
    for all $(x,u) \in \mathcal{W}_{\bar{c}}$ and $\Omega_{\bar{c}}(u) \subseteq X$ for all $u \in [u_{\min},u_{\max}]$.
\end{assumption}

In Proposition~\ref{prop:narrow_feasibility} we show the feasibility of such a metric. Finally, on $X \times [u_{\min},u_{\max}]$ we define the error coordinates $e = s - \bar{s}(u)$, $w = u - v$, and the auxiliary functions:
\begin{align*}
    g(s,u) &= A(u,s,u), \\
    U(s,u) &= \int_{\bar{s}(u)}^{s} g(\sigma,u)d\sigma, \\
    \phi(w,s,u) &= A(u,s,u) - A(u-w,s,u).
\end{align*}
We are now ready to begin providing results. The first lemma relates to the derivative of the acceleration function with respect to the velocity. 

\begin{lemma}\label{lemma:damping}
    Suppose that $u_{\min} > 2T\sqrt{ab}$.
     % Let $w = u - v$ and $\varphi(w,s,u) = A(u,s,u) - A(u-w,s,u)$.
    Then for every $u \in [u_{\min},u_{\max}]$, $s \geq s_0$, $v \in (0,v_0]$,
    \begin{align*}
        \frac{\partial A}{\partial v}(v,s,u) \leq 0,
    \end{align*}
    which guarantees that $w\phi(w,s,u) \leq 0$ for all $w$, with equality only when $w = 0$.
\end{lemma}

\begin{pf}
    Let $f(z) \coloneqq 1 - z^c + cz^c \ln z$ where $c = 2a/\afree(v) > 0$. Then $f(1) = 0$ and $f'(z) = c^2z^{c-1}\ln z \leq 0$ on $(0,1]$, so $f \geq 0$ there. On the branch $z < 1$ we compute
    \begin{align}\label{eq:acceleration_derivative}
        \frac{\partial A}{\partial v} = \afree'(v)f(z) - 2az^{c-1}\frac{1}{s}\frac{\partial s^*}{\partial v},
    \end{align}
    On $z < 1$ it holds that $\afree' \leq 0$ and combined with $f \geq 0$, the first term is less than or equal to $0$. For the second term, when the maximum in~\eqref{eq:desired_gap} is clamped at zero we have that $\partial s^*/\partial v = 0$, while on the unclamped side we get $\partial s^*/\partial v = T + \frac{2v-u}{2\sqrt{ab}}$. On the unclamped side, $v > u - 2T\sqrt{ab}$, $u \geq 2T\sqrt{ab}$ implies $u - 2T\sqrt{ab} \geq \frac{u}{2} - T\sqrt{ab}$, so the negative range is entirely inside the clamped region, where $\partial s^*/\partial v = 0$. Therefore $\partial s^*/\partial v \geq 0$ wherever it is nonzero, and the second term is less than or equal to $0$. Therefore, when $z < 1$ we have that $\partial A/\partial v \leq 0$. For $z \geq 1$, $\partial A/\partial v = -2az(\partial s^*/\partial v)/s \leq 0$. For the second claim, we write
    \begin{align*}
        w\phi(w,s,u) = -w^2\int_0^1 \left(-\frac{\partial A}{\partial v}(u - (1- \theta)w,s,u)\right)d\theta.
    \end{align*}
    From the prior steps, for all values of $\theta$, the integrand is $-\partial A/\partial v \geq 0$, so $w\phi \leq 0$. For $w \neq 0$, note that when $\theta = 1$ the velocity argument is $u$, giving a state of $(u,s,u)$. At this point we have $s^*(u,u) = s_0 + uT > s_0$, so the maximum in~\eqref{eq:desired_gap} is not clamped and we can compute
    \begin{align*}
        \frac{\partial s^*}{\partial v} = T + \frac{u}{2\sqrt{ab}} > 0.
    \end{align*}
    For $z < 1$ this makes the second term in~\eqref{eq:acceleration_derivative} strictly negative while the first term is nonpositive, and for $z \geq 1$ we have $\partial A/\partial v = -2az(\partial s^*/\partial v)/s < 0$. It follows that the integrand is strictly positive at $\theta = 1$, and hence the integral is strictly positive and $w\phi < 0$ for $w \neq 0$.
\end{pf}

We now consider the behavior when there is a fixed input $u \in [u_{\min},u_{\max}]$.

\begin{proposition}\label{prop:frozen_convergence}
    Suppose that the gap increasing condition holds and that $2T\sqrt{ab} < u_{\min} \leq u \leq u_{\max} < v_0$. Then, for any fixed $u \in [u_{\min},u_{\max}]$ every solution starting in $X$ converges to $\bar{x}(u)$.
\end{proposition}

\begin{pf}
    We consider the energy function $E(x,u) = \frac{1}{2}w^2 + U(s,u)$. Computing $\dot{E}$ along solutions of~\eqref{eq:IIDM_follower}, we get 
    \begin{align*}
        \dot{E} = w\dot{w} + \frac{\partial U}{\partial s}\dot{s} + \frac{\partial U}{\partial u}\dot{u}.
    \end{align*}
Since $u$ is constant the last term is equal to zero. For the remaining terms we note that $\frac{\partial U}{\partial s} = g(s,u)$, $\dot{w} = \dot{u} - \dot{v} = -A(u-w,s,u) = -g(s,u) + \phi(w,s,u)$, and $\dot{s} = u - v = w$. Therefore
\begin{align*}
    \dot{E} = w(-g(s,u) + \phi(w,s,u)) + g(s,u)w = w\phi(w,s,u),
\end{align*}
which is less than or equal to zero by Lemma~\ref{lemma:damping}, with equality if and only if $w = 0$. As such, $E$ is a Lyapunov function. We now aim to apply LaSalle's Invariance Principle~\citep{HKK:02}. We have that $X$ is a compact and positively invariant set, $E$ is $C^1$ and bounded below by $0$ on $X$. Therefore, by LaSalle's Invariance Principle every solution converges to the largest positively invariant subset of
\begin{align*}
    \mathcal{S} \coloneqq \{x \in X ~|~ \dot{E} = 0\} &= \{x \in X ~|~ w\phi(w,s,u) = 0 \} \\
    &= \{x \in X ~|~ w = 0\}.
\end{align*}
On $\mathcal{S}$ we have $w = 0$, which gives $v = u$. For a solution to remain in $\mathcal{S}$ we need $\dot{w} = 0$ as well. Calculating gives
\begin{align*}
    \dot{w} = -A(u-w,s,u)\bigg|_{w = 0} = -A(u,s,u) = -g(s,u) = 0.
\end{align*}
This forces $g(s,u) = 0$, which gives $s = \bar{s}(u)$ for an invariant solution. Therefore, the largest invariant subset of $\mathcal{S}$ is $\{x \in X~|~w = 0, s = \bar{s}(u)\} = \{\bar{x}(u)\}$. Therefore, for any fixed input $u$, all solutions beginning in $X$ converge to $\bar{x}(u)$.
\end{pf}

We now show that for fixed inputs all solution will reach an $\epsilon$ neighborhood of $\bar{x}(u)$ in uniform time. For this we consider the $\epsilon$-neighborhood of the equilibrium manifold defined as
    \begin{align*}
        \mathcal{E}_\epsilon = \{x \in X~|~ \norm{x - \bar{x}(u)} \leq \epsilon \text{ for some } u \in [u_{\min},u_{\max}]\}
    \end{align*}
    for $\epsilon > 0$. Let $\epsilon_0$ be the largest value such that $\mathcal{E}_{\epsilon_0}$ lies entirely in the active-braking unsaturated region of $A$, that is, the region where $v < v_0$, $z < 1$, and the maximum in~\eqref{eq:desired_gap} is not clamped at zero.

\begin{lemma}\label{lemma:uniform_entry}
    For every $\epsilon > 0$ there exists $T(\epsilon) < \infty$, depending on $\epsilon$ only through $\min(\epsilon,\epsilon_0)$, such that for every constant $u \in [u_{\min},u_{\max}]$ and every initial condition $x_0 \in X$, the solution satisfies $\norm{x - \bar{x}(u)} \leq \epsilon$ for all $t \geq T(\epsilon)$.
\end{lemma}

\begin{pf}
    Fix $\epsilon' = \min(\epsilon,\epsilon_0) > 0$. On the compact neighborhood $\mathcal{E}_{\epsilon'}$, $K$ is continuous and strictly positive, so there exist local bounds
    % [2026-07-20 edit: removed the analogous D bounds, which were never used in this proof]
    \begin{align*}
        0 &< K_{\min}^{\epsilon'} \leq K(x,u) \leq K^{\epsilon'}_{\max} < \infty
    \end{align*}
    uniform over $\mathcal{E}_{\epsilon'} \times [u_{\min},u_{\max}]$. With these local constants it then holds that
    \begin{align*}
        &\frac{1}{2}\min\{1,K^{\epsilon'}_{\min}\}\norm{x-\bar{x}(u)}^2 \leq E(x,u)\\
        &E(x,u) \leq \frac{1}{2}\max\{1,K_{\max}^{\epsilon'}\}\norm{x-\bar{x}(u)}^2.
    \end{align*}
    Now, define
    \begin{align*}
        c(\epsilon') &= \frac{{\epsilon'}^2}{2}\min\{1,K^{\epsilon'}_{\min}\}\\
        \delta({\epsilon'}) &= \epsilon'\sqrt{\frac{\min\{1,K^{\epsilon'}_{\min}\}}{\max\{1,K^{\epsilon'}_{\max}\}}}.
    \end{align*}
    On $\mathcal{E}_{\epsilon'}$ it then holds that $\norm{x - \bar{x}(u)} \leq \delta(\epsilon') \Rightarrow E(x,u) \leq c(\epsilon') \Rightarrow \norm{x - \bar{x}(u)} \leq \epsilon' \leq \epsilon$. Since $U(s,u)$ is strictly increasing in $|s - \bar{s}(u)|$ and $\frac{1}{2}w^2$ is strictly increasing in $|w|$, $E$ is positive definite in $(w,s-\bar{s}(u))$ on $X$. As such, $\{E \leq c(\epsilon')\}$ is a compact neighborhood contained in $\mathcal{E}_{\epsilon'}$. Then, since $\dot{E} = w\phi \leq 0$ by Lemma~\ref{lemma:damping}, the sublevel set $\{E \leq c(\epsilon')\}$ is positively invariant. It follows that once a solution enters $B_{\delta(\epsilon')}(\bar{x}(u))$ at time $t_1$ it lies in $\{E \leq c(\epsilon')\} \subset \mathcal{E}_{\epsilon'}$ for all $t \geq t_1$.

    Now, by Proposition~\ref{prop:frozen_convergence}, for every $(x_0,u) \in X \times [u_{\min},u_{\max}]$, the frozen input solution satisfies $x(t;x_0,u) \to \bar{x}(u)$. As such, there exists finite $t^*(x_0,u)$ such that $x(t^*(x_0,u);x_0,u) \in B_{\delta(\epsilon')/2}(\bar{x}(u))$\footnote{Here the notation $x(t;x_0,u)$ represents the value $x(t)$ from initial condition $x_0$ with constant input $u$.}. Now, since $F$ is uniformly Lipschitz over $X \times [u_{\min},u_{\max}]$ and $\bar{x}(\cdot)$ is continuous in $u$, the solution map $(x_0,u) \to x(t;x_0,u)$ is continuous in $(x_0,u)$ uniformly on compact time intervals. As such, for each $(x_0,u)$ there exists an open neighborhood $\mathcal{N}(x_0,u) \subset X \times [u_{\min},u_{\max}]$ such that for all $(y_0,u') \in \mathcal{N}(x_0,u)$ it holds that $\norm{x(t^*(x_0,u);y_0,u') - \bar{x}(u')} \leq \delta(\epsilon')$. The collection $\{\mathcal{N}(x_0,u)\}$ then forms a cover for the compact set $X \times [u_{\min},u_{\max}]$. As such, there exists a finite subcover, indexed by $k = 1, \dots, N$, and using this set $T(\epsilon') = \max_{k \in \{1,\dots,N\}} t^*(x_0^k,u^k) < \infty$. Then, every $(x_0,u) \in X \times [u_{\min},u_{\max}]$ is in at least one $\mathcal{N}(x_0^k,u^k)$, so $x(T(\epsilon');x_0,u) \in B_{\delta(\epsilon')}(\bar{x}(u))$.

    Combining the fact that by $T(\epsilon')$ all solutions have entered $B_{\delta(\epsilon')}(\bar{x}(u))$ and they stay in $\mathcal{E}_{\epsilon'}$ for all time afterwards, we have
    \begin{align*}
        \norm{x(t;x_0,u) - \bar{x}(u)} \leq \epsilon' \leq \epsilon \quad \text{for all } t \geq T(\epsilon')
    \end{align*}
    uniformly over all $(x_0,u) \in X \times [u_{\min},u_{\max}]$.
\end{pf}

To conclude the first section of the proof we utilize the prior lemmas to show that a slowly-varying input $u(t) \in \mathcal{U}_\rho$ will converge to the moving ellipsoid.

\begin{proposition}\label{prop:finite_entry_time}
    Assume that the gap-increasing condition holds, $2T\sqrt{ab} < u_{\min} < u_{\max} < v_0$, and Assumption~\ref{assumption:contraction_metric} holds. Let $r = \sqrt{\bar{c}/\lambda_{\max}(P)}$, $T_{acq} = T(r/2)$ from Lemma~\ref{lemma:uniform_entry} and define
    \begin{align*}
        \rho_{acq} = \frac{r/2}{\frac{1}{2}b_F T_{acq}^2 e^{LT_{acq}}+\norm{(1,T)}T_{acq}},
    \end{align*}
    where $L, b_F$ are Lipschitz bounds for $F$ in $x$ and $u$ on $X \times [u_{\min},u_{\max}]$. Then, for every $u \in \mathcal{U}_\rho$ with $\rho \leq \rho_{acq}$ and every $x(0) \in X$, $x(T_{acq}) \in \Omega_{\bar{c}}(u(T_{acq}))$.  
\end{proposition}

\begin{pf}
    Let $y(\cdot)$ solve the frozen problem $\dot{y} = F(y,u(0))$, where $y(0) = x(0)$, both of which stay in $X$ by invariance. Letting $\delta = x - y$, we have that $\norm{\dot{\delta}} \leq L\norm{\delta} + b_F|u(t) - u(0)| \leq L\norm{\delta} + b_F\rho t$. Applying Gr\"{o}nwall's inequality gives $\norm{\delta(t)} \leq b_F\rho \int_0^t se^{L(t-s)}ds \leq \frac{1}{2}b_F\rho t^2 e^{Lt}$. Letting $t = T_{acq}$ gives 
    \begin{align}\label{eq:prop_inequality1}
     \norm{\delta(T_{acq})} = \norm{x(T_{acq}) - y(T_{acq})} \leq \frac{1}{2}b_F\rho T_{acq}^2e^{LT_{acq}}. 
    \end{align}
    Then, by Lemma~\ref{lemma:uniform_entry} with $\epsilon = r/2$, applied to the frozen input $u(0)$ and initial condition $y(0) = x(0) \in X$, we get 
    \begin{align}\label{eq:prop_inequality2}
    \norm{y(T_{acq})-\bar{x}(u(0))} \leq r/2.
    \end{align}
    Finally, since $\bar{x}(u) = (u,s_0 + uT)$ is affine in $u$, we get
    \begin{align}
        \norm{\bar{x}(u(T_{acq})) - \bar{x}(u(0))} &= \norm{(1,T)}|u(T_{acq}) - u(0)| \notag \\ & \leq \norm{(1,T)}\rho T_{acq}, \label{eq:prop_inequality3}
    \end{align}
    since $|u(T_{acq})-u(0)| \leq \rho T_{acq}$ by $u \in \mathcal{U}_\rho$. Now, by combining~\eqref{eq:prop_inequality1},~\eqref{eq:prop_inequality2}, and~\eqref{eq:prop_inequality3} along with the triangle inequality we have
    \begin{align*}
        &\norm{x(T_{acq})-\bar{x}(u(T_{acq})) }\\
        &\leq \norm{x(T_{acq}) - y(T_{acq})} + \norm{y(T_{acq}) - \bar{x}(u(0))} \\ 
        &\qquad + \norm{\bar{x}(u(0)) - \bar{x}(u(T_{acq}))} \\
        &\leq \frac{1}{2}b_F \rho T_{acq}^2e^{LT_{acq}} + \frac{r}{2} + \norm{(1,T)}\rho T_{acq} \\
        &= \frac{r}{2} + \rho\left(\frac{1}{2}b_FT_{acq}^2e^{LT_{acq}} + \norm{(1,T)}T_{acq} \right).
    \end{align*}
    By the definition of $\rho_{acq}$ we get
    \begin{align*}
        \rho\left(\frac{1}{2}b_FT_{acq}^2e^{LT_{acq}} + \norm{(1,T)}T_{acq} \right) \leq \frac{r}{2}.
    \end{align*}
    Combining gives $\norm{x(T_{acq}) - \bar{x}(u(T_{acq}))} \leq r$. By the definition of $r$ we have that the Euclidean ball of radius $r$ is contained in $\Omega_{\bar{c}}(u(T_{acq}))$, completing the proof.
\end{pf}

This completes the first component of the proof of the ESP, in that for all inputs in $\mathcal{U}_\rho$, all trajectories in $X$ will converge to the moving ellipsoid in finite time. We now aim to show that every pair of trajectories inside $\Omega_{\bar{c}}$ contract to each other. The first lemma shows that the ellipsoid is forward invariant.

\begin{lemma}\label{lemma:omega_invariance}
    Let Assumption~\ref{assumption:contraction_metric} hold and define $\ell_P = \norm{(1,T)}_P$. Let $u \in \mathcal{U}_\rho$ and let $x(\cdot)$ be a solution to~\eqref{eq:IIDM_follower} with $x(t) \in \Omega_{\bar{c}}(u(t))$ on some interval $[t_1,t_2)$. Define $\xi(t) = x(t) - \bar{x}(u(t))$ and $W_0(t) = \norm{\xi(t)}_P$. Then for almost every $t$ in $[t_1,t_2)$ $\dot{W}_0 \leq -\lambda W_0 + \ell_P \rho$. Additionally, if $\rho \leq \rho_{inv} \coloneqq \lambda\sqrt{\bar{c}}/\ell_P$, then the moving ellipsoid $\Omega_{\bar{c}}(u(t))$ is forward invariant.
\end{lemma}

\begin{pf}
    First, since $\bar{x}(u)$ is the equilibrium for a frozen input $u$, $F(\bar{x}(u),u) = 0$ for all $u$. Additionally, since $\bar{x}$ is affine in $u$, we get $\frac{d}{dt} \bar{x}(u(t)) = (1,T)^\top \dot{u}(t)$. Taking the derivative $\xi$ along the solution we get
    \begin{align*}
        \dot{\xi} = F(x,u) - F(\bar{x}(u),u) - (1,T)^\top \dot{u}.
    \end{align*}
    Applying the fundamental theorem of calculus to the map $\theta \to F(\bar{x}(u) + \theta,u)$ gives $F(x,u) - F(\bar{x}(u),u) = \left(\int_0^1 J(x_\theta,u)d\theta \right)\xi$, where $x_\theta = \bar{x}(u) + \theta\xi$. By convexity of $\Omega_{\bar{c}}(u)$ the segment from $\bar{x}(u)$ to $x$ lies in the ellipsoid. Now, we differentiate $W_0^2$, giving
    \begin{align*}
        \frac{d}{dt} W_0^2 &= \xi^\top P \dot{\xi} + \dot{\xi}^\top P\xi \\
        &= \xi^\top\left(\int_0^1 [PJ(x_\theta,u) + J(x_\theta,u)^\top P]d\theta \right)\xi \\
        &\qquad - 2\xi^\top P(1,T)^\top \dot{u}.
    \end{align*}
    By Assumption~\ref{assumption:contraction_metric} $S(x_\theta,u) \leq -2\lambda P$ for each $\theta$, so the integral satisfies the bound
    \begin{align*}
        \xi^\top \left(\int_0^1 S(x_\theta,u)d\theta \right)\xi \leq -2\lambda\xi^\top P\xi = -2\lambda W_0^2.
    \end{align*}
    For the second term, applying the Cauchy-Schwarz inequality with the $P$-inner product gives
    \begin{align*}
        |\xi^\top P(1,T)^\top \dot{u}| \leq \norm{\xi}_P \norm{(1,T)^\top}_P |\dot{u}| \leq W_0 \ell_P \rho,
    \end{align*}
    since $|\dot{u}| \leq \rho$. Combining inequalities gives $\frac{d}{dt} W_0^2 \leq - 2\lambda W_0^2 + 2W_0 \ell_P \rho$. When $W_0 > 0$ dividing by $2W_0$ provides the claim. When $W_0(t) = 0$, since $\xi = 0$ the dynamics satisfies $\norm{\dot{\xi}}_P \leq \ell_P |\dot{u}| \leq \ell_P \rho$. As such, the Dini derivative, $D^+$ satisfies $D^+ W_0(t) \leq \ell_P \rho = -\lambda W_0(t) + \ell_P \rho$. Applying the comparison lemma provides the claim.

    Now we show forward invariance. Suppose that $x(t_1) \in \Omega_{\bar{c}}(u(t_1))$, so $W_0(t_1) \leq \sqrt{\bar{c}}$. Evaluating $\dot{W}_0 \leq -\lambda W_0 + \ell_P \rho$ on the boundary $W_0 = \sqrt{\bar{c}}$ we get
    \begin{align*}
        \dot{W}_0\bigg|_{W_0 = \sqrt{\bar{c}}} \leq -\lambda \sqrt{\bar{c}} + \ell_P \rho.
    \end{align*}
    If $\rho \leq \rho_{inv} = \lambda\sqrt{\bar{c}}/\ell_P$, then $\ell_P\rho \leq \lambda\sqrt{\bar{c}}$, so $\dot{W}_0|_{W_0 = \sqrt{\bar{c}}} \leq 0$.
Since the derivative is non-positive at the boundary, trajectories cannot cross outwards past $W_0 = \sqrt{\bar{c}}$. As such, for all $t \geq t_1$ we have $x(t) \in \Omega_{\bar{c}}(u(t))$, giving forward invariance.
\end{pf}

The final intermediate step is to prove that the trajectories inside $\Omega_{\bar{c}}$ contract to each other.

\begin{lemma}\label{lemma:omega_contraction}
    Suppose Assumption~\ref{assumption:contraction_metric} holds and define $g_P = \max_{\mathcal{W}_{\bar{c}}} \norm{G(x,u)}_P$, where $G(x,u) = \partial F/\partial u$. Let $u_1,u_2 \in \mathcal{U}_{\rho}$ and let $x_1,x_2$ be solutions satisfying $x_i(t) \in \Omega_{\bar{c}}(u_i(t))$ for all $t \geq T_0$, $i = 1,2$.
    % \begin{align*}
    %     x_i(t) \in \Omega_{\bar{c}}(u_i(t)) \qquad \forall t \geq T_0,~i = 1,2.
    % \end{align*}
    Then, for all $t \geq T_0$
    \begin{align*}
        \norm{x_1(t)-x_2(t)}_P &\leq e^{-\lambda(t-T_0)}\norm{x_1(T_0) - x_2(T_0)}_P \\
        &+ g_P\int_{T_0}^{t}e^{-\lambda(t-s)} |u_1(s) - u_2(s)|ds.
    \end{align*}
    In particular, the integral term is bounded above by $\frac{g_P}{\lambda}\sup_{T_0 \leq \tau \leq t}|u_1(\tau) - u_2(\tau)|$. As such, for a shared input $u_1 = u_2$, the trajectories contract exponentially. If $u_1(t) - u_2(t) \to 0$, then $x_1(t) - x_2(t) \to 0$.
\end{lemma}

\begin{pf}
     We denote $\delta x(t) = x_1(t) - x_2(t)$, $\delta u(t) = u_1(t) - u_2(t)$, and $W(t) = \norm{\delta x(t)}_P$. Let $x_\theta(t) = \theta x_1(t) + (1-\theta)x_2(t)$ and $u_\theta(t) = \theta u_1(t) + (1-\theta)u_2(t)$ for $\theta \in [0,1]$. Write $x_i = \bar{x}(u_i) + \xi_i$, where $\xi_i^\top P \xi_i \leq \bar{c}$. Since $\bar{x}$ is affine, we can write $\theta\bar{x}(u_1) + (1-\theta)\bar{x}(u_2) = \bar{x}(u_\theta)$. Additionally, note that $\theta\xi_1 + (1-\theta)\xi_2$ lies in the ellipsoid of radius $\sqrt{\bar{c}}$ by the convexity of the set. As such,
     \begin{align*}
         x_\theta = \bar{x}(u_\theta) + \theta \xi_1 + (1- \theta)\xi_2 \in \Omega_{\bar{c}}(u_\theta).
     \end{align*}
     Therefore $(x_\theta,u_\theta) \in \mathcal{W}_{\bar{c}}$ for all $\theta \in [0,1]$. Since $F$ is $C^1$ in $\Omega_{\bar{c}}$, and as such along the segment, we apply the fundamental theorem of calculus to get $F(x_1,u_1) - F(x_2,u_2) = \int_0^1 \frac{d}{d\theta} F(x_\theta,u_\theta)d\theta = \left(\int_0^1 J(x_\theta,u_\theta)d\theta \right)\delta x +\left(\int_0^1 G(x_\theta,u_\theta)d\theta \right)\delta u,$
     % in which $J(x,u) = \partial F/\partial x$ and $B(x,u) = \partial F/\partial u$. 
     Since $\dot{\delta x} = \dot{x}_1 - \dot{x}_2$, it is equal to the above sum of integrals. Now, taking the derivative of $W$, we get
     \begin{align*}
         \frac{d}{dt}W^2 &= \delta x^\top \left(\int_0^1 S(x_\theta,u_\theta)d\theta\right)\delta x + 2\delta x^\top P\left(\int_0^1 Gd\theta\right)\delta u \\
         &\leq -2\lambda W^2 + 2Wg_P |\delta u|.
     \end{align*}
     Therefore $\dot{W} \leq -\lambda W + g_P|\delta u|$ when $W > 0$. A direct application of the comparison lemma in its integral form yields the first result, and the supremum bound follows since $\int_{T_0}^{t} e^{-\lambda(t-s)}ds \leq 1/\lambda$. For the second claim, suppose $u_1(t) - u_2(t) \to 0$. For arbitrary $\epsilon > 0$ select $T_\epsilon \geq T_0$ such that for all $t \geq T_\epsilon$ $|u_1(t) - u_2(t)| \leq \epsilon$. Applying the comparison lemma from $T_\epsilon$ gives
     \begin{align*}
         \norm{x_1(t) - x_2(t)}_P \leq e^{-\lambda(t-T_{\epsilon})}\norm{x_1(T_{\epsilon}) - x_2(T_{\epsilon})}_P + \frac{g_P\epsilon}{\lambda}.
     \end{align*}
     As both trajectories lie in $\Omega_{\bar{c}}$ for all $t \geq T_0$, we have that $\norm{x_1(T_{\epsilon}) - x_2(T_{\epsilon})} \leq 2\sqrt{\bar{c}}$. Therefore, $\limsup_{t \to \infty} \norm{x_1(t) - x_2(t)}_P \leq g_P \epsilon/\lambda$. Since this holds for all $\epsilon > 0$, $x_1(t) - x_2(t) \to 0$.
\end{pf}
Using these intermediate results we are able to show that the single-follower IIDM satisfies the ESP on $X \times \mathcal{U}_\rho$.
\begin{theorem}\label{thrm:ESP_slow}
    Consider the single-follower IIDM dynamics~\eqref{eq:IIDM_follower}. Suppose that the gap-increasing condition and Assumption~\ref{assumption:contraction_metric} hold, $\rho \leq \rho^* = \min\{\rho_{inv},\rho_{acq}\}$, and $x(0) \in X$. Then, this dynamics satisfies the ESP over the set $X \times \mathcal{U}_\rho$.
\end{theorem}
\begin{pf}
    Let $x_1$ and $x_2$ be solutions of~\eqref{eq:IIDM_follower} driven by input $u \in \mathcal{U}_\rho$. By Proposition~\ref{prop:finite_entry_time} there exists a time $T$ such that both solutions satisfy $x_i(T) \in \Omega_{\bar{c}}(u(T))$. Then, by Lemma~\ref{lemma:omega_invariance} both solutions stay in $\Omega_{\bar{c}}(u(t))$ for all $t \geq T$. Since both solutions are in $\Omega_{\bar{c}}(u(t))$, Lemma~\ref{lemma:omega_contraction} gives
    \begin{align*}
        \norm{x_1(t)-x_2(t)}_P \leq e^{-\lambda(t-T)}\norm{x_1(T)-x_2(T)}_P
    \end{align*}
    for all $t \geq T$. As such the solutions converge to each other exponentially, satisfying Definition~\ref{def:echo-state-property}.
\end{pf}
As such the single follower IIDM satisfies the ESP. Following the construction of the metric certifying Assumption~\ref{assumption:contraction_metric} (see Appendix B), we get that this can hold for $\rho \approx 0.02$, which limits us to signals that vary significantly slower than actual vehicles. However, in Section~\ref{sec:simulations} we illustrate successful RC prediction for signals with larger derivatives, suggesting that the ESP holds for a more general class of input signals. Now that we have shown that the single follower IIDM vehicle satisfies the ESP, we want to show that this extends to a platoon of $N$ vehicles, which is given by the following result. The inductive argument is related to string stability analysis for vehicle platoons~\citep{DS-JKH:96,SF-YZ-SEL-ZC-HXL-LL:19}, in that bounds are propagated down the cascade of vehicles.

\begin{proposition}\label{prop:platoon_ESP}
    Consider an $N$-vehicle platoon in which each following vehicle is defined by the dynamics~\eqref{eq:IIDM_follower}. Define the rate-propagation gain by
    \begin{align*}
    \Theta = \frac{L_{A}\ell_P}{\lambda \sqrt{\lambda_{\min}(P)}} = L_A \kappa,
    \end{align*}
    where $L_{A}$ is the Lipschitz constant of $A(\cdot,\cdot,\cdot)$ with respect to $x$, and $\lambda,P$ are as defined in Assumption~\ref{assumption:contraction_metric}. For $j \in \{0,\dots,N-1\}$ let $\rho^{(j)} = \Theta^j\rho$, and suppose that $\hat{\rho} \coloneqq \rho \max\{1,\Theta,\dots,\Theta^{N-1}\} < \rho^*$. Then any two solutions for the platoon, $x(t), \tilde{x}(t)$, driven by the same input $u_0 \in \mathcal{U}_\rho$ satisfying $u_0(t) \in [u_{\min}+N\kappa\hat{\rho}, u_{\max}-N\kappa\hat{\rho}]$, converge to each other, that is $\norm{x^{(i)}(t) - \tilde{x}^{(i)}(t)} \to 0$ for every $i \in \{1,\dots,N\}$. Further, this convergence is uniform, so the platoon satisfies the ESP.
\end{proposition}
\begin{pf}
        We proceed by induction on the vehicle index. For $j \in \{1,\dots,N\}$ let $\Pi_j(t) = (1 + (t-T^{(j)})^{j-1})e^{-\lambda(t-T^{(j)})}$, where the uniform times $T^{(j)}$ are constructed below. We show that for each $i \in \{1,\dots,N\}$: (i) for every $\eta > 0$, eventually $|\dot{v}^{(i)}(t)| \leq \rho^{(i)} + \eta$; (ii) for every $\eta > 0$, eventually $v^{(i)}(t) \in [u_{\min}+(N-i)\kappa\hat{\rho} - \eta,\, u_{\max}-(N-i)\kappa\hat{\rho} + \eta]$; and (iii) there exists $C_i > 0$, independent of the pair of solutions, such that $\norm{x^{(i)}(t) - \tilde{x}^{(i)}(t)}_P \leq C_i\Pi_i(t)$ for all $t \geq T^{(i)}$. Interpreting $v^{(0)} = u_0$, properties (i) and (ii) hold at index $0$ by assumption, with rate $\rho^{(0)} = \rho$ and margin $N\kappa\hat{\rho}$. By (i) and (ii) at index $i-1$, the input $u^{(i-1)} \coloneqq v^{(i-1)}$ to vehicle $i$ eventually takes values in $[u_{\min}+\kappa\hat{\rho}-\eta, u_{\max}-\kappa\hat{\rho}+\eta] \subset [u_{\min},u_{\max}]$ (taking $\eta < \kappa\hat{\rho}$) with rate bound $\rho^{(i-1)} + \eta \leq \hat{\rho} + \eta < \rho^*$ for $\eta$ sufficiently small. Therefore Proposition~\ref{prop:finite_entry_time} and Lemma~\ref{lemma:omega_invariance} apply to vehicle $i$, so there exists a time $T^{(i)}$ after which $x^{(i)}(t) \in \Omega_{\bar{c}}(v^{(i-1)}(t))$ and $\tilde{x}^{(i)}(t) \in \Omega_{\bar{c}}(\tilde{v}^{(i-1)}(t))$.

        We first propagate the rate bound (i). Since $A$ vanishes at the equilibrium $\bar{x}(u^{(i-1)})$ of the input $u^{(i-1)} = v^{(i-1)}$, Lipschitzness of $A$ gives
        \begin{align}\label{eq:induction_1}
            |\dot{v}^{(i)}(t)| = |A(x^{(i)},u^{(i-1)})| &\leq L_A \norm{x^{(i)} - \bar{x}(u^{(i-1)})} \notag \\
            &\leq \frac{L_A}{\sqrt{\lambda_{\min}(P)}}W_0(t),
        \end{align}
        with $W_0(t) = \norm{x^{(i)} - \bar{x}(v^{(i-1)}(t))}_P$. Applying Lemma~\ref{lemma:omega_invariance} along with the comparison lemma with the input rate $\rho^{(i-1)} + \eta$ gives
        \begin{align}\label{eq:induction_2}
            W_0(t) \leq W_0(t_1)e^{-\lambda(t - t_1)} + \frac{\ell_P(\rho^{(i-1)}+\eta)}{\lambda}.
        \end{align}
        Substituting~\eqref{eq:induction_2} into~\eqref{eq:induction_1} gives 
        \begin{align*}
            |\dot{v}^{(i)}(t)| \leq \frac{L_A W_0(t_1)}{\sqrt{\lambda_{\min}(P)}}e^{-\lambda(t-t_1)} + \Theta (\rho^{(i-1)}+\eta),
        \end{align*}
        and since $W_0(t_1) \leq \bar{W} \coloneqq \max_{X \times [u_{\min},u_{\max}]}\norm{x - \bar{x}(u)}_P$, which is finite, for every $\eta' > 0$ it holds that $|\dot{v}^{(i)}(t)| \leq \Theta\rho^{(i-1)} + \eta' = \rho^{(i)} + \eta'$ within a uniform time, establishing (i) at index $i$.

        For the range bound (ii), converting~\eqref{eq:induction_2} into the Euclidean norm gives that, eventually, $|v^{(i)}(t) - v^{(i-1)}(t)| \leq \kappa(\rho^{(i-1)}+\eta) + \eta \leq \kappa\hat{\rho} + \eta'$ for every $\eta' > 0$. Combined with (ii) at index $i-1$, this establishes (ii) at index $i$.

        Finally we establish the contraction bound (iii). For $i = 1$ the two copies share the input $u_0$, so applying Lemma~\ref{lemma:omega_contraction} from $T^{(1)}$ gives $\norm{x^{(1)}(t) - \tilde{x}^{(1)}(t)}_P \leq 2\sqrt{\bar{c}}e^{-\lambda(t-T^{(1)})} \leq C_1 \Pi_1(t)$ with $C_1 = \sqrt{\bar{c}}$. For $i \geq 2$, property (iii) at index $i-1$ gives, for $s \geq T^{(i-1)}$,
        \begin{align*}
            |v^{(i-1)}(s) - \tilde{v}^{(i-1)}(s)| \leq \frac{C_{i-1}}{\sqrt{\lambda_{\min}(P)}}\Pi_{i-1}(s).
        \end{align*}
        Applying Lemma~\ref{lemma:omega_contraction} from $T^{(i)}$ with this input difference gives, for $t \geq T^{(i)}$,
        \begin{align*}
            \norm{x^{(i)}(t) - \tilde{x}^{(i)}(t)}_P &\leq 2\sqrt{\bar{c}}e^{-\lambda(t-T^{(i)})} \\
            &+ \frac{g_P C_{i-1}}{\sqrt{\lambda_{\min}(P)}}\int_{T^{(i)}}^{t}\!e^{-\lambda(t-s)}\Pi_{i-1}(s)ds.
        \end{align*}
        A direct computation bounds the integral by a polynomial of degree $i-1$ in $(t-T^{(i)})$ multiplied by $e^{-\lambda(t-T^{(i)})}$, so there exists $C_i$, depending only on $\bar{c}$, $g_P$, $\lambda$, $\lambda_{\min}(P)$, and $C_{i-1}$, such that $\norm{x^{(i)}(t) - \tilde{x}^{(i)}(t)}_P \leq C_i\Pi_i(t)$ for all $t \geq T^{(i)}$, establishing (iii). In particular $v^{(i)}(t) - \tilde{v}^{(i)}(t) \to 0$, closing the induction.

        The times $T^{(i)}$ are uniform over pairs of solutions, being built from the uniform acquisition time of Proposition~\ref{prop:finite_entry_time} and the uniform rate-settling times above. Letting $D_X$ denote the diameter of $X$ and defining
        \begin{align*}
            \delta_i(t) = \begin{cases}
                D_X & t < T^{(i)} \\
                \frac{C_i}{\sqrt{\lambda_{\min}(P)}}\Pi_i(t) & t \geq T^{(i)},
            \end{cases}
        \end{align*}
        we have $\norm{x^{(i)}(t) - \tilde{x}^{(i)}(t)} \leq \delta_i(t)$ for all $t \geq 0$, and $\delta_i(t) \to 0$ as $t \to \infty$ since the polynomial factor is dominated by the exponential decay. Therefore $\norm{x(t) - \tilde{x}(t)} \leq \delta(t) \coloneqq \sum_{i=1}^N \delta_i(t)$ for every pair of solutions, with $\delta(t) \to 0$, satisfying Definition~\ref{def:echo-state-property} and showing that the platoon satisfies the ESP.
\end{pf}

Proposition~\ref{prop:platoon_ESP} illustrates the ESP for a reservoir formed by a platoon of $N$-vehicles when the input is the actual velocity of the lead vehicle. However, a more realistic input corresponds with using the desired velocity of the lead vehicle, $v_0^{(1)}$. The argument used in Proposition~\ref{prop:platoon_ESP} can be used to show that this method also works, provided that the lead velocities from two initial conditions converge to each other. In Lemma~\ref{lemma:e1_decay} we provide a proof showing that this holds, which allows for the use of the desired velocity as an input.

%--------------------------------------------------------------------
\section{Simulations}\label{sec:simulations}
%--------------------------------------------------------------------
In the prior section we illustrated that the IIDM dynamics satisfies the ESP for a class of slowly-varying input signals, and as such IIDM-RC is theoretically valid for those signals. In this section we provide numerical simulations for more physically realistic inputs, illustrating that IIDM-RC can be used in physically realistic cases. We first show that velocities and gaps of a simulated traffic network can be predicted to a high degree of accuracy using IIDM-RC, which can then be extrapolated to position prediction. For these simulations we use multiple traffic models for the simulated network being predicted. This is to ensure that success is not dependent only on using IIDM dynamics for both the target and reservoir. Following this we compare the use of IIDM-RC for prediction with using both an LSTM network and an ESN.

The task considered in the simulations is as follows. We consider a network of $5$ vehicles\footnote{Detailed network parameter values used in this section are provided in the Appendix D.}, governed by a car-following dynamics and traveling on a single-lane circular road. The road is assumed to be large enough such that when combined with the velocity limits on the vehicles, the lead vehicle will not interact with the trailing vehicle of the network. In order to allow for heterogeneity in the network values are selected from a uniform random variable centered around a selected value. This network is given an input signal as the desired velocity of the lead vehicle, resulting in a non-constant system trajectory. We assume that this network is undersensed, and we are able to determine the velocity, or associated gap, of one or a small number of vehicles. The goal is then to predict the gaps and/or velocities of all of the vehicles in the network. Based on the gap or velocity predictions, the position of the vehicles relative to their starting location can be determined.

\subsection{Prediction with IIDM-RC}
In this section we show the use of IIDM-RC for predicting the gaps and velocities of the target vehicle network. We use the gap of the third vehicle, $s_3$, as the input to the reservoir, after mapping it into a desired velocity by the function $u(t) = \frac{s_3}{4} + 20$. The input signal to the target network for these simulations is $v_{0}^{(1)}(t) = 4\cos(\frac{2t}{15}) + 6\sin(\frac{2t}{12}) + 25$ m/s, which has a range between $[15,35]$ m/s and a maximum derivative of $1.53$ m/s$^2$. 

Figure~\ref{fig:rc_predictions} shows the velocity and gap predictions for the second and fourth vehicles of a target network governed by the IIDM dynamics. In both cases a brief transient settles to an accurate prediction within the first thirty seconds. The transient stems from the initial conditions of the target network rather than the reservoir: training uses a washout period, so the trained attractor reflects behavior near the equilibrium trajectory of the input, and predictions started far from this trajectory show a larger transient.

\begin{figure}[ht]
\centering
  \includegraphics[width = 0.85\linewidth]{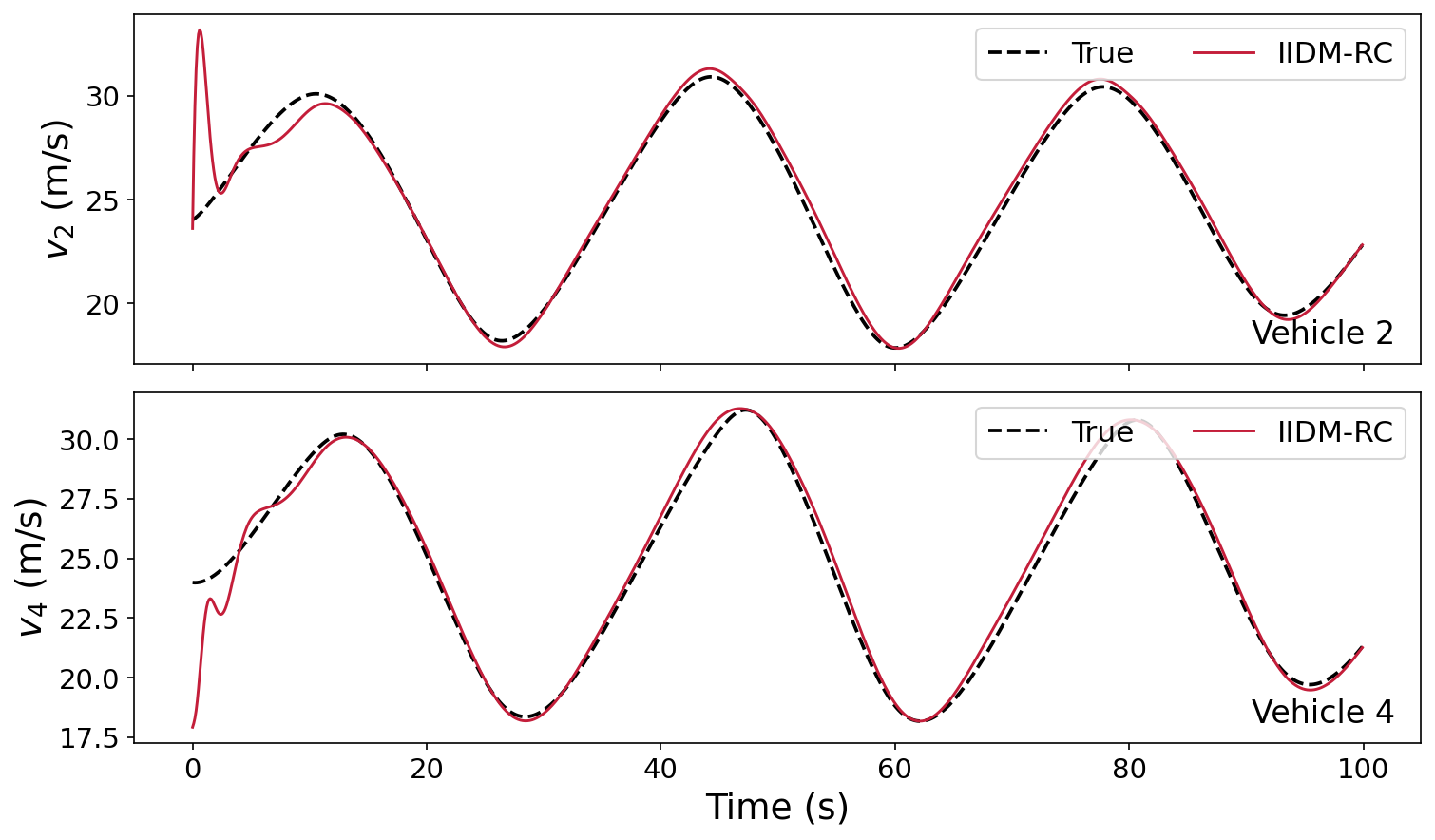}\\
  \includegraphics[width = 0.85\linewidth]{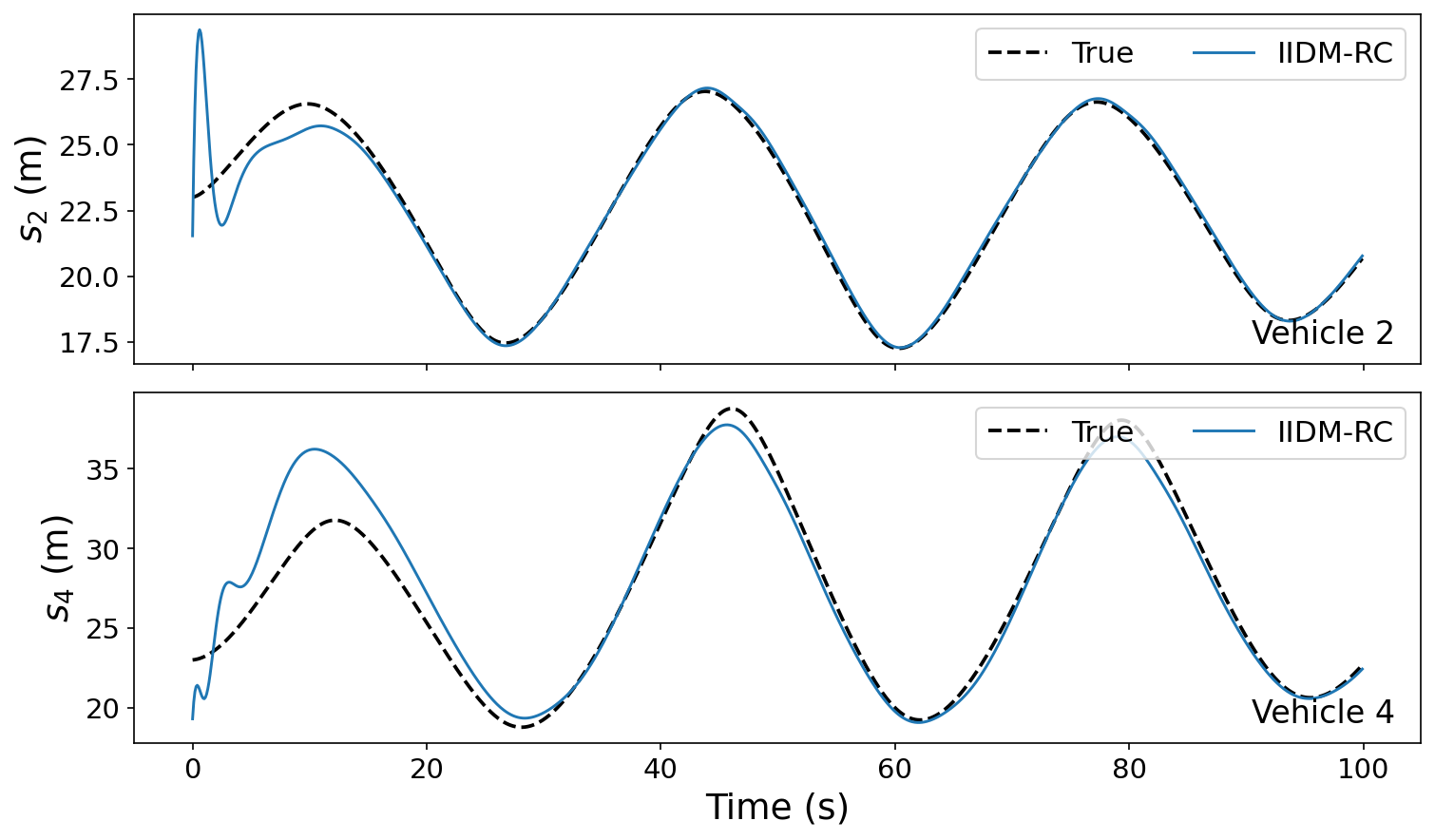}
  \caption{Velocity (top) and gap (bottom) predictions for the second and fourth vehicles in the $5$ vehicle IIDM target network with the third gap as the reservoir input. The brief initial transient arises from target initial conditions outside the trained attractor; once trajectories approach the attractor the prediction becomes accurate.}\label{fig:rc_predictions}
\end{figure}

In order to ensure that the IIDM-RC prediction is valid for traffic networks in general we check the performance prediction for multiple car-following models. In particular, we use the same target network defined by IIDM, IDM, and OVM dynamics, each with parameters set to exhibit similar behavior. Table~\ref{tab:nrmse_diff_targets} shows the prediction errors in the settled regime using IIDM-RC for each of these target networks. The prediction for all three models have a high level of prediction performance (given as mean per-vehicle NRMSE\footnote{Normalization is with respect to the standard deviation of the true signal. That is, for signal $y$ and prediction $\hat{y}$ $NRMSE = \frac{\sqrt{\frac{1}{T}\sum_{t=1}^T (\hat{y}_t - y_t)}}{\sigma_{y}}$, where $\sigma_y$ is the standard deviation.}), which illustrates that the IIDM-RC is able to accurately predict general traffic models, not only networks defined with IIDM dynamics.

\begin{table}[h]
\centering
\caption{Mean per-vehicle NRMSE  for prediction of both gaps and velocities in IIDM, OVM, and IDM target networks using IIDM-RC with reservoir input $s_3$. Values are for $t > 30$ seconds and are computed over $20$ seeds.}
\label{tab:nrmse_diff_targets}
\begin{tabular}{lcc}
\hline
Target & Gap NRMSE & Velocity NRMSE \\
\hline
IIDM & $0.0687 \pm 0.0055$ & $0.0654 \pm 0.0077$ \\
OVM  & $0.0372 \pm 0.0072$ & $0.0780 \pm 0.0033$ \\
IDM  & $0.0390 \pm 0.0086$ & $0.0614 \pm 0.0094$ \\
\hline
\end{tabular}
\end{table}
The gap and velocity predictions can then be used to predict the positions of the vehicles over time. If the velocity of the vehicles is predicted then the position can be calculated through integrating the differential equation $\dot{x} = v$ (or through a discrete approach). If the gaps are predicted, and we know the velocity of the sensed vehicle, we can compute the actual position of the sensed vehicle directly, before computing the other positions through the sum of the gaps and vehicle lengths\footnote{A video of the position prediction using the gaps as the predicted value can be seen at \url{https://ece.uwaterloo.ca/~sl2smith/position_IIDM-RC_prediction.mp4}}. We note that the prediction of positions through the gaps is safer than through the velocities. Since all gap predictions are positive, and the ground truth of one vehicle is known, the predicted positions will not collide. Conversely, if the velocity is used, errors can add up over time, leading to the incorrect prediction of a collision.

\subsection{Comparison with Other Methods}
In this section we compare the performance of the IIDM-RC with other machine learning methods for the gap prediction of the undersensed IIDM network. In particular, we will compare with the use of an LSTM network and an ESN. For all three methods we use an IIDM target network and optimize over the hyperparameters to maximize performance. We compare both the prediction error and training parameters for each method to identify benefits and drawbacks of each method.

For this comparison we predict an undersensed $5$ vehicle network with an IIDM-RC composed of $10$ vehicles, an ESN with $100$ nodes, and an LSTM with two hidden layers of $64$ nodes each. These sizes were picked over hyperparameter sweeps that illustrated further increasing the size did not result in increased prediction performance and details are given in Appendices~\ref{app:param_values} and~\ref{app:lstm_size_comparison}. In Table~\ref{tab:nrmse_method_comparison} we compare the performance of the three methods over a $100$ second prediction window for the prediction of the vehicle gaps with the third gap being provided as the input. As the third gap is known it is excluded from the provided prediction values and mean calculation. The table reports the settled window ($30$--$100$\,s); over the full horizon ($0$--$100$\,s) the mean per-vehicle NRMSE is $0.0729$ for IIDM-RC, $0.1111$ for ESN-100, and $0.1082$ for the LSTM. 

% ── Table 2: NRMSE comparison ─────────────────────────────────────────────────
% [2026-07-21 shrink: reduced to settled window and single column; full-horizon means moved to text]
\begin{table}[h]
\centering
\caption{Settled-window ($30$--$100$\,s) gap-prediction NRMSE for IIDM-RC, ESN-100, and LSTM $2\times64$ on the IIDM target with signal $v_0^{(1)}(t) = 4\cos(2t/15) + 6\sin(2t/12)+25$, averaged over 20 seeds. Bold indicates the lowest error in each row.}
\label{tab:nrmse_method_comparison}
\small\setlength{\tabcolsep}{4pt} % [2026-07-21 shrink: fit single column]
\begin{tabular}{lccc}
\hline
Metric & IIDM-RC & ESN & LSTM \\
\hline
Mean per-vehicle NRMSE & \textbf{0.0726} & 0.0997 & 0.0857 \\
% \addlinespace
Vehicle 1   & \textbf{0.0544} & 0.0751 & 0.0696 \\
Vehicle 2   & \textbf{0.0444} & 0.0624 & 0.0450 \\
Vehicle 4   & 0.1109 & 0.1282 & \textbf{0.1087} \\
Vehicle 5   & \textbf{0.0808} & 0.1333 & 0.1197 \\
\hline
\end{tabular}
\end{table}

Table~\ref{tab:nrmse_method_comparison} shows that for both the full horizon and settled windows the IIDM-RC achieves a similar level of accuracy to the other two methods, with the lowest error on every vehicle on all but vehicle $4$ in the settled window. The minor improvement of the IIDM-RC compared to the other methods can be understood based on the internal workings of each method. While the LSTM and ESN approaches are both successful in their predictions, the system being predicted appears as a generic time-series input, as there is no relation between the target system dynamics and the model. However, with the IIDM-RC, the physics of the target system matches with that of the reservoir, allowing for the encoding of similar trajectories in the reservoir and target network. 

While the IIDM-RC has a minor benefit in prediction accuracy, it is more interesting to consider the training time differences between the methods. Table~\ref{tab:training_comparison} provides a comparison of the training times and parameters between the IIDM-RC, ESN, and LSTM networks in the gap prediction task. From this table we see that the IIDM-RC is significantly faster than both the ESN and LSTM approaches, due to a combination of the training method and the fewer parameters. The fact that both the IIDM-RC and ESN approaches are trained with a single singular value decomposition (SVD) allows for near instantaneous training. Meanwhile the LSTM must go through 100 epochs of training, leading to a significant slowdown. This training time difference is particularly significant in the case of a changing target network. Should the prediction system need to be retrained if using the IIDM-RC predictor this takes less than one millisecond, compared to the $24$ second timeframe of the LSTM. This creates a long window in which one is able to receive predictions from an IIDM-RC but not an LSTM. Due to the similar levels of prediction accuracy, this training time mismatch is a significant benefit of using an IIDM-RC approach to predicting behavior in an undersensed traffic network.

% ── Table 3: Training cost comparison ─────────────────────────────────────────
\begin{table}[h] % [2026-07-21 shrink: was table*]
\centering
\caption{Comparison of training costs for IIDM-RC, ESN-100, and LSTM
$2\times64$. Times exclude target-system simulation and data collection;
IIDM-RC and ESN times are averaged over 100 runs, LSTM over 10 runs.}
\label{tab:training_comparison}
\footnotesize\setlength{\tabcolsep}{3.5pt} % [2026-07-21 shrink: fit single column]
\begin{tabular}{lccc}
\hline
Metric & IIDM-RC & ESN & LSTM \\
\hline
Training time    & $0.29 \pm 0.20$\,ms     & $8.3 \pm 0.4$\,ms          & $24.6 \pm 0.1 $\,s               \\
Training steps   & 3000           & 3000               & $3000\times100$  \\
Parameters       & 55             & 505                & 50{,}757                \\
Algorithm        & SVD            & SVD                & Adam                    \\
Slowdown         & $1\times$   & $\approx29\times$ & $\approx 85{,}000 \times$ \\
\hline
\end{tabular}
\end{table}

\section{Conclusion}
We considered the problem of using a traffic network within physical reservoir computing for the purpose of predicting an undersensed traffic network. We show that the traffic network satisfies the echo state property for a class of slowly-varying inputs when using the IIDM to simulate traffic dynamics. While the slowly-varying requirement is included within the mathematical analysis, numerical simulations suggest that the ESP holds for a wider class of signals, including those with accelerations within a physically meaningful range. The simulations compare the performance of IIDM-RC with more traditional ESNs and LSTMs, and show that IIDM-RC achieves a similar level of prediction accuracy, while having significantly lower training time and computational cost. As such, while IIDM-RC is not a general replacement for other methods, it is appealing when the conditions of the undersensed network change rapidly. Directions for future work include utilizing the prediction of the undersensed network state to control its behavior and increasing the complexity of the vehicular reservoir by considering networks with multiple roads and lanes.
% ----------------------------------------------------------------------
\appendix
% ----------------------------------------------------------------------

\counterwithin{theorem}{section}
\counterwithin{figure}{section}
\counterwithin{table}{section}

\section{Auxiliary Result}\label{app:math_results}

The result provided is used in proving that trajectories remain in the set $X$.
\begin{proposition}\longthmtitle{Invariance of a Lower-Bound}\label{prop:lower_bound_invariance}
  Let $\map{f}{\R}{\R}$ be continuous and differentiable. Suppose there exists $t_0$ such that if $t \geq t_0$ and $f(t) \leq c$, then $f'(t) > 0$. Then, if $f(t_0) \geq c$, $f(t) \geq c$ for all $t \geq t_0$. In addition, $f(t) > c$ for all $t > t_0$.
\end{proposition}
\begin{pf}
  By contradiction suppose that $f(t) < c$ for some $t > t_0$. Then, by continuity and compactness there exists a final time $t^* \in [t_0,t]$ such that $f(t^*) \geq c$, after which $f < c$ on $(t^*,t]$. By the mean value theorem this gives $t' \in (t^*,t)$ with $f'(t') = (f(t) - f(t^*))/(t - t^*)$. Since $f(t) \leq c \leq f(t^*)$ this derivative is negative. But $f(t') < c$ and $t' > t_0$, which contradicts the hypothesis that gives $f'(t') > 0$.

  To see the second part of the claim we note that if $f(t) \geq c$ for all $t > t_0$, it is a local minimum on $[t_0,\infty)$. Suppose at time $t$ the trajectory reaches $f(t) = c$. As a local minimum this means that $f'(t) = 0$, which contradicts the fact that $f'(t) > 0$ if $f(t) \leq c$. Therefore, for all $t > t_0$, $f(t) > c$.
\end{pf}  

\section{Feasibility of the Contraction Metric}\label{app:metric_feasibility}

Here we show that Assumption~\ref{assumption:contraction_metric} is feasible for a narrow input range. From there the result can be certified on the a given range through testing of the convex hull.
\begin{proposition}\label{prop:narrow_feasibility}
    Fix $\bar{u}_0 \in [u_{\min},u_{\max}]$ and write $K_0 = K(\bar{x}(\bar{u}_0),\bar{u}_0)$, $D_0 = D(\bar{x}(\bar{u}_0),\bar{u}_0)$. Choose
    \begin{align*}
        c = \frac{D_0}{2}, \qquad \alpha = K_0 + cD_0, \qquad P = \begin{bmatrix} 1 & -c \\ -c & \alpha\end{bmatrix} \succ 0.
    \end{align*}
    Then, $S(\bar{x}(\bar{u}_0),\bar{u}_0) = \diag(-D_0,-D_0K_0) \prec 0$. As such, there exists $\lambda_0 > 0$ such that $S(\bar{x}(\bar{u}_0),\bar{u}_0) \preceq -2\lambda_0 P$. Further, by continuity of the map $(x,u) \to S(x,u)$ on $X \times [u_{\min},u_{\max}]$ there exists $\bar{c} > 0$ and a neighborhood $I_0$ containing $\bar{u}_0$ such that Assumption~\ref{assumption:contraction_metric} holds with $\lambda = \lambda_0/2$ when $[u_{\min},u_{\max}] \subseteq I_0$.
\end{proposition}

\begin{pf}
    We compute $\det P = \alpha - c^2 = K_0 + c(D_0-c) > 0$, so $P \succ 0$. Then, computing $PJ$, we get
    \begin{align*}
        PJ = \begin{bmatrix}
        1 & -c \\ -c & \alpha
        \end{bmatrix}\begin{bmatrix}
            -D & K \\ -1 & 0
        \end{bmatrix} = \begin{bmatrix}
            -D+c & K \\
            cD - \alpha & -cK
        \end{bmatrix}.
    \end{align*}
    At the point $(\bar{x}(\bar{u}_0),\bar{u}_0)$,  the off-diagonal entry of $S$ is equal to $K_0 + cD_0 - \alpha = 0$ and $S = \diag(-2(D_0-c),-2cK_0) = \diag(-D_0,-D_0K_0) \prec 0$ since both $D_0$ and $K_0$ are positive. Since $S \prec 0$ and $P \succ 0$ at the equilibrium there exists $\lambda_0$ satisfying $S(\bar{x}(\bar{u}_0),\bar{u}_0) \preceq -2\lambda_0 P$. Explicitly, we get $2\lambda_0P + S \preceq 0$, which holds when the eigenvalues of $P^{-1}(-S/2)$ are greater than $\lambda_0$. Since $-S/2 = \diag(D_0/2,D_0K_0/2) \succ 0$ and $P \succ 0$, this is feasible with
    \begin{align*}
        \lambda_0 = \frac{1}{2} \lambda_{\min}(P^{-1}(-S/2)),
    \end{align*}
    or the largest $\lambda_0$ such that $\diag(D_0,D_0K_0) - 2\lambda_0P \succeq 0$. From this value we have that $S(\bar{x}(\bar{u}_0),\bar{u}_0) + \lambda_0 P = (S(\bar{x}(\bar{u}_0),\bar{u}_0) + 2\lambda_0P) - \lambda_0 P \preceq 0 -\lambda_0P = -\lambda_0 P \prec 0$. The map $(x,u) \mapsto S(x,u) + \lambda_0 P$ is continuous and the set
    \begin{align*}
        \mathcal{O} = \{(x,u) ~|~ S(x,u) + \lambda_0 P \prec 0\}
    \end{align*}
    is open on $X \times [u_{\min},u_{\max}]$ and contains $\bar{x}(\bar{u}_0),\bar{u}_0)$. Since $\mathcal{O}$ is open and contains $(\bar{x}(\bar{u}_0),\bar{u}_0)$ there exists $\bar{c} > 0$ and neighborhood $I_0$ such that $\mathcal{W}_{\bar{c}} \subset \mathcal{O}$ whenever $[u_{\min},u_{\max}] \subseteq I_0$ since $W_{\bar{c}} \to \{(\bar{x}(\bar{u}_0),\bar{u}_0)\}$ as $\bar{c} \to 0$ and $[u_{\min},u_{\max}] \to \{\bar{u}_0\}$.
\end{pf}

This shows that there is a feasible metric on any small input range. To verify the existence of a feasible metric for a larger input range we take a numerical approach. We note that~\eqref{eq:certificate} is linear in $J$ and affine in $(K,D)$. It follows that the certificate holds if and only if it holds at the extreme points of the convex hull of pairs $(K,D)$ attained on $\mathcal{W}_c$. 

\section{Convergence of Lead Vehicle Velocities}\label{app:lead_vehicle_convergence}
Here we illustrate that two copies of a vehicle platoon governed by IIDM dynamics with different initial conditions that the lead vehicle velocity will converge under the same desired velocity. This allows for the use of the desired velocity of the lead vehicle in the platoon as the reservoir input.

We recall that for a single IIDM lead vehicle we treat it as if there is no additional leader, so the dynamics are governed solely by the free acceleration function. That is,
\begin{align*}
    \dot{v}(v_0) = \afree(v,v_0) = \begin{cases}
            a\!\left[1 - \left(\dfrac{v}{v_0}\right)^{\!\delta}\right] & v \leq v_0,\\[8pt]
    -b\!\left[1 - \left(\dfrac{v_0}{v}\right)^{\!a\delta/b}\right] & v > v_0.
    \end{cases}
\end{align*}

To show that the lead vehicle velocities converge we use the error dynamics and take input $u = v_0$. Let $v_1$ and $v_2$ be the velocities of two copies of the platoon with different initial conditions. The error is given by $e = v_1 - v_2$, and the error dynamics are $\dot{e} = \dot{v}_1 - \dot{v}_2$. Convergence of the trajectories then corresponds to decay of the error dynamics to zero, which we illustrate below.

\begin{lemma}\longthmtitle{Exponential Decay of Lead Vehicle Error}\label{lemma:e1_decay}
  Suppose that $v(0) \in [u_{\min},u_{\max}]$. Then, the error $e$ satisfies
  \begin{align}\label{eq:e1_decay_equation}
    |e(t)| \leq |e(0)|e^{-\kappa_1 t},
  \end{align}
  where $\kappa_1 = \min_{v\in [u_{\min},u_{\max}]} (-\frac{\partial \afree(v,u)}{\partial v})$.
\end{lemma}

\begin{pf}
  Computing the partial derivative of $\afree(v,u)$ with respect to $v$ in each regime gives $-\frac{a\delta}{u^\delta}v^{\delta - 1}$ when $v \leq u$ and $-a\delta v_0^{a\delta/b}v^{-a\delta/b - 1}$ when $v > u$. From these expressions we see that $\kappa_1 > 0$. As such, $\afree(v,u)$ is strictly decreasing in $v$ on $[u_{\min},u_{\max}]$, and applying the mean-value theorem gives $e\dot{e} \leq -\kappa_1 e^2$. A direct application of Gr\"onwall's Inequality yields the result. 
\end{pf}

Therefore, since the error decays exponentially, we have that the velocities of the lead vehicles for any two copies of a platoon with the desired velocity as the input converge. This is equivalent to the ESP for the lead vehicle, and when combined with the argument made in Proposition~\eqref{prop:platoon_ESP} allows for use of the desired velocity as the input for the IIDM-RC while retaining the ESP for the $N$-vehicle platoon.
\section{Simulation Parameters}\label{app:param_values}
In this section we provide the network and reservoir parameters for the simulations of Section~\ref{sec:simulations}.  In Tables~\ref{tab:iidm_params}-\ref{tab:ovm_params} we provide the parameters for the IIDM, IDM, and OVM target networks. In these tables $\mathcal{U}[x,y]$ represents a uniform random distribution on the interval $[x,y]$.

% ── IIDM target ──────────────────────────────────────────────────────────────
\begin{table}[h]
\centering
\caption{IIDM target system parameters.}
\label{tab:iidm_params}
\begin{tabular}{lcc}
\hline
Parameter & Value & Units \\
\hline
$N$        & 5                             & --      \\
$L$        & 400                           & m       \\
Vehicle Length     & 1                             & m       \\
$s_0$      & 4.0                           & m       \\
$v_0$      & 40.0                          & m/s     \\
$\delta$ & 4 & -- \\
$a$        & $\mathcal{U}[1.0,\,1.8]$     & m/s$^2$ \\
$b$        & $\mathcal{U}[1.0,\,1.8]$     & m/s$^2$ \\
$T$        & $\mathcal{U}[0.5,\,1.25]$    & s       \\
\hline
\end{tabular}

\end{table}

\begin{table}[h]
\centering
\caption{IDM target system parameters.}
\label{tab:idm_params}
\begin{tabular}{lcc}
\hline
Parameter & Value & Units \\
\hline
$N$        & 5                               & --       \\
$L$        & 400                             & m        \\
Vehicle Length     & 1                               & m        \\
$s_0$      & 4.0                             & m        \\
$v_0$      & 40.0                            & m/s      \\
$\delta$   & 4                               & --       \\
$a$        & $\mathcal{U}[1.15,\,1.65]$     & m/s$^2$  \\
$b$        & $\mathcal{U}[1.15,\,1.65]$     & m/s$^2$  \\
$T$        & $\mathcal{U}[0.875,\,1.075]$   & s        \\
\hline
\end{tabular}

\end{table}

% ── OVM target ────────────────────────────────────────────────────────────────
\begin{table}[h]
\centering
\caption{OVM target system parameters.}
\label{tab:ovm_params}
\begin{tabular}{lcc}
\hline
Parameter & Value & Units \\
\hline
$N$        & 5                               & --        \\
$L$        & 400                             & m         \\
Vehicle Length     & 1                               & m         \\
$s_0$      & 4.0                             & m         \\
$v_0$      & 40.0                            & m/s       \\
$\tau$     & $\mathcal{U}[1.0,\,1.1]$       & s         \\
$h_c$      & $\mathcal{U}[13.5,\,16.5]$     & m         \\
$\kappa$   & $\mathcal{U}[0.35,\,0.45]$     & m$^{-1}$  \\
\hline
\end{tabular}

\end{table}

The parameters used for the ESN and LSTM networks are given in Tables~\ref{tab:esn_params} and~\ref{tab:lstm_params}, respectively. The chosen ESN parameters are based on a sweep of both network size and spectral radius, with the smallest effective network size being taken. The LSTM parameters were swept over both layer size and number of hidden layers. Appendix~\ref{app:lstm_size_comparison} discusses the results over these parameter sweeps.

% ── IIDM reservoir ──────────────────────────────────────────────────────────────
\begin{table}
\caption{IIDM-RC reservoir parameters.}
\label{tab:reservoir_params}
\centering
\begin{tabular}{lcc}
\toprule
Parameter & Value & Units \\
\midrule
$N_r$   & 10                      & --      \\
$L$     & 250                     & m       \\
$s_0$   & 2.0                     & m         \\
$v_0$   & 40.0                    & m/s     \\
$\delta$ & 4 & -- \\
$a$     & $\mathcal{U}[1.0, 1.8]$ & m/s$^2$ \\
$b$     & $\mathcal{U}[1.0, 1.8]$ & m/s$^2$ \\
$T$     & $\mathcal{U}[0.5, 1.25]$ & s      \\
$\beta$ & $10^{-5}$               & --      \\
\bottomrule
\end{tabular}
\end{table}

\begin{table}[h]
\centering
\caption{ESN hyperparameters.}
\label{tab:esn_params}
\begin{tabular}{ll}
\toprule
Parameter & Value \\
\midrule
Reservoir size $N$               & 100 \\
Spectral radius $\rho$           & 0.7 \\
Connection density               & 0.10 \\
Input scaling                    & 0.1 \\
Leaking rate $\alpha$            & 1.0 \\
Bias                             & 0.1 \\
Tikhonov $\beta$                 & 0.01 \\
Input signal                     & $s_3$ (normalised) \\
\bottomrule
\end{tabular}
\end{table}

\begin{table}[h]
\centering
\caption{LSTM hyperparameters.}
\label{tab:lstm_params}
\begin{tabular}{ll}
\toprule
Parameter & Value \\
\midrule
Layers                           & 2 \\
Hidden units per layer           & 64 \\
Input sequence length            & 20 steps (2\,s) \\
Training epochs                  & 100 \\
Optimiser                        & Adam \\
Learning rate                    & 0.001 \\
Weight decay                     & 0.001 \\
Batch size                       & 64 \\
Dropout (between layers)         & 0.1 \\
Input signal                     & $s_3$ (sliding window) \\
\bottomrule
\end{tabular}
\end{table}

\section{LSTM Size Comparison}\label{app:lstm_size_comparison}
Here we provide plots illustrating that performance of the LSTM is mostly constant across different sizes. To see this we tested the prediction of the undersensed IIDM traffic using LSTM networks with varying numbers of layers and hidden layer sizes. We consider between $2$ and $4$ hidden layers and between $64$ and $1024$ units in each layer. In Figure~\ref{fig:lstm_layer_comparison} we provide comparisons between the different sized LSTM networks and the IIDM-RC prediction. We see that for all four vehicles the NRMSE is relatively constant across the LSTM networks and is greater than the IIDM-RC prediction. What this shows is that the prediction error is not a function of not having a large enough network but instead is due to being unable to fully encode the target physics. The similarity between the physics of the IIDM-RC reservoir and the target network then explains the better performance of the IIDM-RC.

\begin{figure*}
\centering
    \includegraphics[width=0.9\linewidth]{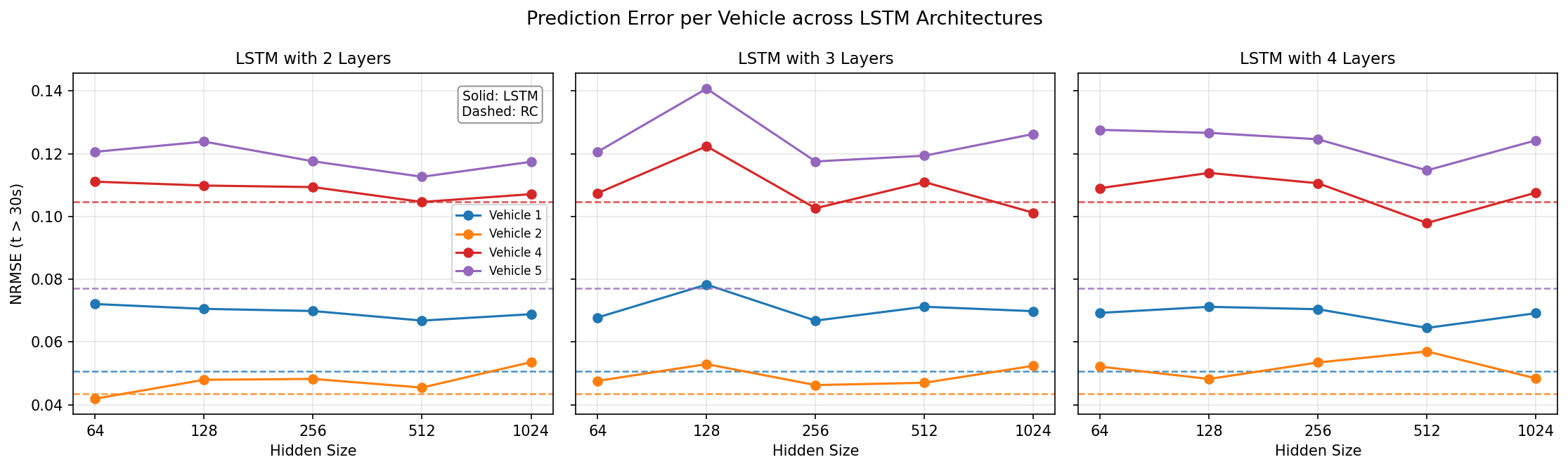}
    \caption{Comparison of the prediction NRMSE in the settled time domain for the undersensed $5$ vehicle network. Networks with sizes ranging between $2\times64$ and $4\times1024$ are shown on a per-vehicle basis with the input vehicle (vehicle 3) omitted. The dashed lines provide comparison values from using IIDM-RC for the same prediction. We see in general changing the size of the LSTM has minimal impact on the prediction error and the IIDM-RC reliably outperforms the LSTM.}\label{fig:lstm_layer_comparison}
\end{figure*}
% ----------------------------------------------------------------------

\end{document}